\documentclass[11pt,letterpaper]{article}
\usepackage{amsmath,amssymb,amsfonts,amsthm}
\usepackage{tikz}
\usepackage{enumitem}
\usepackage{bbm}
\usepackage{xcolor}
\usepackage{caption}
\usepackage[ruled,vlined]{algorithm2e}
\usepackage[margin=1in]{geometry}
\usepackage[colorlinks, citecolor = blue, linkcolor = magenta]{hyperref}
\usepackage{cleveref}
\usepackage{float}

\newcommand{\bflam}{\boldsymbol{\lambda}}
\newcommand{\bflamp}{\boldsymbol{\lambda'}}
\numberwithin{equation}{section}
\newtheorem{theorem}{Theorem}[section]
\newtheorem{lemma}[theorem]{Lemma}
\newtheorem{corollary}[theorem]{Corollary}
\theoremstyle{definition}
\newtheorem{remark}{Remark}
\newtheorem{proposition}[theorem]{Proposition}
\newtheorem{definition}[theorem]{Definition}
\newtheorem{problem}{Problem}[section]
\newtheorem{claim}{Claim}
\newtheorem*{claim*}{Claim}
\allowdisplaybreaks

\renewcommand{\epsilon}{\varepsilon}

\title{Two-State Spin Systems with Negative Interactions\thanks{A preliminary short version of this paper, without Sections \ref{sec:hardness} or \ref{sec:five}, appeared in the proceedings of the conference ITCS 2024. For the purpose of Open Access, the authors have applied a CC BY public copyright licence to any Author Accepted Manuscript version arising from this submission. All data is provided in full in the results section of this paper.}}

\author{Yumou Fei\thanks{School of Mathematical Sciences, Peking University} \and Leslie Ann Goldberg \thanks{Department of Computer Science, University of Oxford} \and Pinyan Lu \thanks{Shanghai University of Finance and Economics; Key Laboratory of Interdisciplinary Research of Computation and Economics (SUFE), Ministry of Education; supported by National Key R\&D Program of China (2023YFA1009500)} }

\date{13 June 2025}

\begin{document}
\maketitle

\begin{abstract}
We study the approximability of computing the partition functions of two-state spin systems. The problem is parameterized by a $2\times 2$ symmetric matrix. Previous results on this problem were restricted either to the case where the matrix has non-negative entries, or to the case where the diagonal entries are equal, i.e. Ising models. In this paper, we study the generalization to arbitrary $2\times 2$ interaction matrices with real entries. We show that in some regions of the parameter space, it's \#P-hard to even determine the sign of the partition function, while in other regions there are fully polynomial approximation schemes for the partition function. Our results reveal several new computational phase transitions.
\end{abstract}

\section{Introduction}
Spin systems are widely studied in statistical physics, probability theory and theoretical computer science. They can express many natural graph invariants such as the number of independent sets or the number of $k$-colorings, as well as spin models of statistical physics such as the Ising model or the Potts model.

\subsection{The Problem}\label{subsec:problem}
The partition function of a $q$-state spin system can be parameterized by a symmetric matrix $A\in \mathbb{R}^{q\times q}$. It associates with every graph $G=(V,E)$ the real number
$$Z(G;A)=\sum_{\sigma\in [q]^{V}}\prod_{\{u,v\}\in E}A_{\sigma(u),\sigma(v)}.$$

\begin{remark}
Throughout the paper, the word ``graph'' refers to undirected multigraph permitting self-loops and parallel edges.
\end{remark}

Fixing a symmetric matrix $A$, the complexity of exactly computing $Z(G;A)$ given input $G$ was studied and settled by \cite{dyer2000complexity} (for $A$ with $0/1$ entries), \cite{bulatov2005complexity} (for $A$ with nonnegative entries), \cite{goldberg2010complexity} (for $A$ with real algebraic entries), and \cite{cai2013graph} (for $A$ with complex algebraic entries). They proved the remarkable ``dichotomy theorem'', which states that either computing $Z(G;A)$ can be done in polynomial time or it is \#P-hard, and the class of tractable matrices $A$, although lacking a simple explicit characterization, is polynomial-time decidable. 

In this paper, we study the problem of \textit{approximately} computing $Z(G;A)$. For simplicity of handling models of computation, we restrict our attention to rational numbers. We will deal exclusively with two-state spin systems ($q=2$), as they already appear challenging enough:

\begin{problem}\label{prob:main}
For which symmetric matrices $A=\begin{bmatrix}
A_{00} & A_{01}\\ A_{10} & A_{11}\end{bmatrix}\in\mathbb{Q}^{2\times 2}$ is approximately computing $Z(G;A)$ tractable?
\end{problem}

If $A_{01}=A_{10}=0$, it is easy to see that $Z_G$ can be computed exactly in polynomial time (see also \cite{bulatov2005complexity}). In the following, assume $A_{01}=A_{10}\neq 0$, and we normalize the matrix $A$ so that $A_{01}=A_{10}=1$. Then $A$ is given by two parameters $A_{00}=\beta$ and $A_{11}=\gamma$. Whenever $\beta$ and $\gamma$ are fixed, we abbreviate $Z(G;A)$ to $Z_{G}$.

\Cref{prob:main} is well studied for nonnegative matrix entries. In the nonnegative quadrant $\beta,\gamma\geq 0$, 
\cite{goldberg2003computational} gave an FPRAS for the ``ferromagnetic'' case $\beta\gamma\geq 1$. The ``antiferromagnetic'' case $\beta\gamma<1$ was later very much settled by a series of work \cite{goldberg2003computational, weitz2006counting, sly2010computational, sly2012computational, li2013correlation, sinclair2014approximation, galanis2016inapproximability}. They proved a computational phase transition that coincides with the boundary of the ``uniqueness region'' (uniqueness of Gibbs measure on infinite regular trees). Their results in fact extend much beyond \Cref{prob:main}: the computational phase transition for the anti-ferromagnetic case holds even when external fields are allowed. 

However, much less is known about \Cref{prob:main} when $\beta$ or $\gamma$ is negative. The only existing results in this direction are about the Ising model, which means the special case $\beta=\gamma$. Embedded in a broader study about Tutte polynomials, the following 
theorems from \cite{goldberg2014complexity} and \cite{goldberg2007inapproximability} classified the approximation complexity of Ising partition functions with negative $\beta$:   
\begin{proposition}[Corollary 28 of \cite{goldberg2014complexity}]\label{prop:isingsharpP}
Fix rational numbers $\beta,\gamma$ such that $\beta=\gamma\in(-1,0)$. It is \#P-hard to determine the sign of the partition function $Z_{G}$, given an input graph $G$.
\end{proposition}
\begin{proposition}[Lemma 7 of \cite{goldberg2007inapproximability}]\label{prop:PMequivalent}
Fix rational numbers $\beta,\gamma$ such that $\beta=\gamma<-1$. Approximating the partition function $Z_{G}$ for an input graph $G$ is
equivalent to approximately counting perfect matchings
in general graphs 
in the sense that there are approximation-preserving reductions between these problems, implying that either both problems have an FPRAS or neither problem has an FPRAS. 
Whether approximately counting perfect matchings is tractable or not is a central open question in the area. 
\end{proposition}

Note that at the point $(\beta,\gamma) = (-1,-1)$, $Z_G$ can be computed exactly in polynomial time ($Z_{G}$ is $2^{|V(G)|}$ if all vertex degrees are even and 0 otherwise).

\subsection{Our Results}\label{subsec:results}

In this paper, we explore \Cref{prob:main} in the case $\min\{\beta,\gamma\}<0$. In \Cref{sec:hardness}, we will prove the following generalization of \Cref{prop:isingsharpP}:
\begin{theorem}\label{thm:sharpPhard}
Fix rational numbers $\beta,\gamma$ such that $\min\{\beta,\gamma\}<0$ and $-2<\beta+\gamma <1$, but $(\beta,\gamma)\not\in\{(1,-1),(-1,1)\}$. It is \#P-hard to determine the sign of the partition function $Z_{G}$, given an input graph $G$.
\end{theorem}

Of course Theorem~\ref{thm:sharpPhard} has ramifications for the complexity of approximating~$Z_G$. In particular, an FPRAS for approximating~$Z_G$ gives a polynomial-time randomised algorithm for computing the sign of~$Z_G$, which is not possible assuming that \#P-hard problems cannot be solved in randomised polynomial time.

Note that when $(\beta,\gamma) \in \{(1,-1),(-1,1)\}$,
$Z_G$ can be computed exactly in polynomial time~\cite[Theorem 1.2]{goldberg2010complexity}.

It is then of great interest to find whether the two lines $\beta+\gamma=-2$ and $\beta+\gamma=1$ are actual thresholds of approximation complexity. The following two theorems, both of which will be proved in \Cref{sec:approx}, show that the former line is indeed an actual threshold:
\begin{theorem}\label{thm:FPTAS}
Fix rational numbers $\beta,\gamma$ such that $\beta\neq \gamma$ and $|\beta+\gamma|>2$. For any positive integer $\Delta$, there is an FPTAS for $Z_{G}$, where $G$ is an input graph of maximum degree no more than $\Delta$ (without the bounded degree requirement, there is a quasi-polynomial time approximation scheme).
\end{theorem}
\begin{theorem}\label{thm:FPRAS}
Fix rational numbers $\beta,\gamma$ such that $\beta\neq \gamma$ and $|\beta+\gamma|\geq 2$. There is an FPRAS for $Z_{G}$, where $G$ is an input graph.
\end{theorem}

Note that \Cref{thm:FPRAS} contains the boundary case $|\beta+\gamma|=2$, which \Cref{thm:FPTAS} doesn't. What's more, since \Cref{thm:FPRAS} doesn't require the input graph to be bounded degree, it is not subsumed by \Cref{thm:FPTAS} even for the range $|\beta+\gamma|>2$. 

The algorithm of \Cref{thm:FPTAS} is based on the zero-freeness framework of \cite{barvinok2016combinatorics} and Asano's contraction method \cite{asano1970theorems}, while the algorithm of \Cref{thm:FPRAS} relies on the ``windability'' framework of \cite{mcquillan2013approximating} and a holographic transformation. The zero-freeness framework, achieving notable successes in problems with nonnegative parameters (e.g. \cite{peters2019conjecture}), applies naturally in the presence of mixed signs as well. In contrast, the ``windability'' framework, or more generally Markov-chain-based methods only make sense for problems with positive parameters. It is thus somewhat surprising that, via a holographic transformation, we are able to transform the problem into one with positive parameters and furthermore prove the rapid mixing of a Markov chain, for the \textit{maximum possible} parameter range based on a lower bound on $|\beta+\gamma|$. 

Now, the obvious challenge is to determine the approximation complexity in the remaining region, that is, for parameters $\beta,\gamma$ such that $\min\{\beta,\gamma\}<0$ and $1\leq\beta+\gamma<2$. Unfortunately, we are unable to fully achieve this goal. Instead, we give some results that might provide some insights into this challenge (see \Cref{sec:concluding} for more discussion). 

\begin{theorem}\label{thm:positivity}
Let $\beta,\gamma$ be real numbers such that $\beta+\gamma\geq 1$. Then for any graph $G$, the partition function $Z_{G}$ is positive.
\end{theorem}

\begin{remark}
For $\beta+\gamma\leq -2$, it is easy to find a graph $G$ such that $Z_{G}<0$ (e.g. a single self-loop or a triangle). When $-2<\beta+\gamma<1$ and $\min\{\beta,\gamma\}<0$ and $\beta,\gamma\not\in\{(-1,1),(1,-1)\}$, \Cref{thm:sharpPhard} implies that $Z_{G}$ is negative for some graph $G$. When $(\beta,\gamma)\in\{(1,-1),(-1,1)\}$, $Z_{G}$ is negative for $G=K_{4}$ (the 4-clique). Combined with these observations, \Cref{thm:positivity} completely determines the range of parameters $\beta$ and $\gamma$ for which the partition function $Z_{G}$ is always nonnegative: the union of the half plane $\beta+\gamma\geq 1$ and the first quadrant $\beta,\gamma\geq 0$.  
\end{remark}

\Cref{thm:positivity} suggests that approximating the partition function is unlikely to be \#P-hard when $\beta+\gamma\geq 1$, and hence the line $\beta+\gamma=1$ is likely some threshold of approximation complexity.

The proof of \Cref{thm:positivity} is by induction on the size of the graph and will be given in \Cref{subsec:positivity}. In fact, such recursion methods have also been widely used to show zero-freeness of some partition functions on the complex plane (e.g. \cite{liu2022correlation}), which in turn leads to deterministic approximation algorithms by the framework of \cite{barvinok2016combinatorics}. For our partition function, we show in \Cref{subsec:circular} that such recursions can be used to determine the largest zero-free disk around 0 for the range $\{(\beta,\gamma):\gamma<0\text{ and }1\leq \beta+\gamma\leq 2\}$:
\begin{theorem}\label{thm:diskaround0}
Let $\beta,\gamma$ be real numbers such that $\gamma<0$ and $1\leq \beta+\gamma\leq 2$. Then for any graph $G$, the polynomial $Z_{G}(x)$ as defined in \Cref{subsec:notations} is zero-free on the disk $\left\{z\in\mathbb{C}:|z|<\frac{\beta-1}{1-\gamma}\right\}$. Furthermore, $\frac{\beta-1}{1-\gamma}$ is the maximum possible radius such that the zero-freeness holds for all graphs $G$.
\end{theorem}
Using the same type of recursion in a more sophisticated way, we are able to show that the partition function $Z_{G}$ is efficiently computable if $\beta+\gamma$ is sufficiently close to 2, by slightly extending the zero-free region of \Cref{thm:diskaround0}. This suggests the line $\beta+\gamma=2$ is \textit{not} really a computational threshold: 

\begin{theorem}\label{thm:notthreshold}
Let $g:(1,+\infty)\rightarrow (0,1)$ be the following function:
\begin{equation}
g(\beta)=\max\left\{\frac{\beta-2}{\beta^{2}-1},\frac{(\beta-1)^{2}}{\beta^{3}+\beta^{2}-\beta}\right\}.
\end{equation}
Fix rational numbers $\beta,\gamma$ such that $\min\{\beta,\gamma\}<0$ and $\beta+\gamma>2-g(\max\{\beta,\gamma\})$. For any positive integer $\Delta$, there is an FPTAS for $Z_{G}$, where $G$ is an input graph of maximum degree no more than $\Delta$ (without the bounded degree requirement, there is a quasi-polynomial time approximation scheme). 
\end{theorem}

\Cref{thm:notthreshold} breaks the algorithmic barrier $\beta+\gamma=2$ presented by \Cref{thm:FPTAS} and shows that the line $\beta+\gamma=2$ behaves in a completely different way from the line $\beta+\gamma=-2$. The proof of \Cref{thm:notthreshold} will be given in \Cref{subsec:uncentered}.

\begin{figure}[H]
    \centering
    \includegraphics[scale=2]{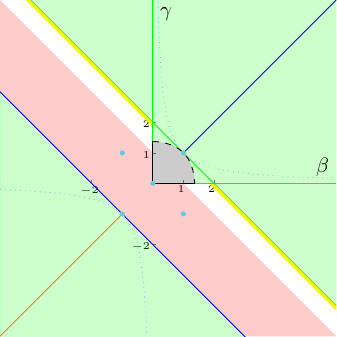}
    \captionsetup{singlelinecheck=off}
    \caption{An illustration of the complexity classification (schematic; not to scale).
    The sky-blue dots $\{(\beta,\gamma):\beta\gamma=1\}\cup\{(-1,1),(0,0),(1,-1)\}$ are where $Z_{G}$ can be computed exactly in polynomial time \cite{bulatov2005complexity,goldberg2010complexity}.
    Sitting in the bottom-left corner of the first quadrant, the black region is where approximating the partition function is known to be NP-hard \cite{sly2012computational}. The dashed line stands for the uniqueness boundary for anti-ferromagnetic 2-spin systems.
    When $(\beta,\gamma)$ falls in the green regions, there is an FPTAS for $Z_{G}$ on bounded degree graphs (due to \Cref{thm:FPTAS} and \cite{li2013correlation}), and an FPRAS for $Z_{G}$ on all graphs (due to \Cref{thm:FPRAS} and \cite{chen2022rapid}).
    The thin yellow strips to the left of the $\beta+\gamma=2$ line are where an FPTAS for bounded degree graphs is given by \Cref{thm:notthreshold}, suggesting that $\beta+\gamma=2$ is not a threshold. 
    When $(\beta,\gamma)$ falls on the blue lines, there is an FPRAS for $Z_{G}$ (the line $\beta+\gamma=-2$ follows from \Cref{thm:FPRAS}, while the ray $\beta=\gamma>1$ is due to \cite{jerrum1993polynomial}). 
    In the red region, apart from the points $(-1,1)$ and $(1,-1)$, approximating $Z_{G}$ is \#P-hard (\Cref{thm:sharpPhard}). 
    On the orange line, approximating the partition function is equivalent to approximately counting perfect matchings \cite{goldberg2007inapproximability}.} 
    \label{fig:results}
\end{figure}

\subsection{More Related Work}

Most of the literature studying 2-state spin systems is restricted to the case where the edge interactions $\beta$ and $\gamma$ and the vertex weights $\lambda$ (i.e. external fields, see \Cref{subsec:notations}) are all nonnegative. But there are also some related lines of work where negative or even complex parameters have received more attention. 

For instance, in the case of the Ising model, besides the results mentioned in \Cref{prop:isingsharpP} and \Cref{prop:PMequivalent},  \cite{goldberg2017complexity} studies the approximation complexity of $Z(G;A)$, where $A=\begin{bmatrix}\beta & 1\\
1 & \beta\end{bmatrix}$ and $\beta$ is any algebraic \textit{complex} number, partly motivated by the connection with quantum complexity classes.

Another line of research concerns the hard-core model (this corresponds to interactions $\beta=1$ and $\gamma=0$ with external fields). Regarding this model there has been much work on the complexity of approximating $Z_{G}(\lambda)$ for bounded-degree graphs $G$ varying parameter $\lambda\in\mathbb{C}$~\cite{harvey2018computing, galanis2017inapproximability, bezakova2019inapproximability}. Here the study of the complexity of approximation is intimately related to the study of optimal zero-free regions of the polynomial $Z_{G}(x)$~\cite{bencs2021limit, de2021zeros, bencs2023complex}.
The techniques used in Section~\ref{sec:hardness} are directly analogous to those in~\cite{bezakova2019inapproximability} --- see Remark~\ref{remfive}.

\section{Preliminaries}

As in \Cref{subsec:problem}, we consider a fixed symmetric matrix $A=\begin{bmatrix}
A_{00} & A_{01}\\ A_{10} & A_{11}\end{bmatrix}\in\mathbb{Q}^{2\times 2}$. 

\subsection{Notations}\label{subsec:notations}
For $G=(V,E)$ and $\bflam\in \mathbb{R}^{V}$, let 
$$Z_{G}(\bflam)=\sum_{\sigma\in\{0,1\}^{V}}\left(\prod_{\{u,v\}\in E}A_{\sigma(u),\sigma(v)}\prod_{v\in V}\lambda_{v}^{\sigma(v)}\right).$$
Here $\bflam$ is the vector of \textit{external fields}.
As a special case, we have $Z_{G}=Z_{G}(\mathbf{1})$. By setting $\lambda_{v}=x$ for all $v\in V$, we get a univariate polynomial $Z_{G}(x)$.

For $v\in V$, let $[Z_{G,v}(\bflam)]$ be a $2\times 1$ vector (we call it the \emph{activity vector} of $v$ in $G$) whose $i$-th coordinate is 
$$\left[Z_{G,v}(\bflam)\right]_{i}=\sum_{\sigma\in\{0,1\}^{V}}\mathbbm{1}\{\sigma(v)=i\}\left(\prod_{\{u,v\}\in E}A_{\sigma(u),\sigma(v)}\prod_{v\in V}\lambda_{v}^{\sigma(v)}\right).$$
When $\left[Z_{G,v}(\bflam)\right]_{0}\neq 0$, we define the ratio $R_{G,v}(\bflam)=\left[Z_{G,v}(\bflam)\right]_{1}/\left[Z_{G,v}(\bflam)\right]_{0}$.

For $u,v\in V$, let $\left[Z_{G,u,v}(\bflam)\right]$ be a $2\times 2$ matrix whose $(i,j)$ entry is
$$\left[Z_{G,u,v}(\bflam)\right]_{i,j}=\sum_{\sigma\in\{0,1\}^{V}}\mathbbm{1}\{\sigma(u)=i\}\mathbbm{1}\{\sigma(v)=j\}\left(\prod_{\{u,v\}\in E}A_{\sigma(u),\sigma(v)}\prod_{v\in V}\lambda_{v}^{\sigma(v)}\right).$$

\subsection{\#CSP and Holant Problems}

The problem of computing the partition function of a spin system can be seen as an instance of $\mathsf{\#CSP}$ problem with a single symmetric binary constraint function. In fact, we may identify the symmetric matrix $A$ with the binary function $\psi$ defined by $\psi(i,j)=A_{ij}$. Then we can denote by $\mathsf{\#CSP}(\{\psi\})$ the problem of computing $Z(G;A)$ given $G$.

In \Cref{subsec:prelimFPRAS,subsec:proofFPRAS}, we will utilize the connection between \#CSP problems and Holant problems. A Holant instance is a graph $G=(V,E)$ with a variable on each edge and a constraint on each vertex. The constraint on a vertex $v$ is a function $F_{v}:\{0,1\}^{J_{v}}\rightarrow\mathbb{C}$, where $J_{v}$ is the set of edges incident to $v$.

\begin{remark}
Self-loops might bring in some ambiguity here. But in this paper, we don't consider self-loops in the context of Holant problems, as we're not going to need them.
\end{remark}

Let $\mathcal{F}$ be a class of constraint functions. A Holant problem $\mathsf{Holant}(\mathcal{F})$ asks for computing the partition function
$$\sum_{\sigma\in \{0,1\}^{E}}\prod_{v\in V}F_{v}(\sigma_{|J_{v}})$$
on input $(G,(F_{v})_{v\in V})$, where each $F_{v}\in\mathcal{F}$.

A particular family of constraint functions we will use in \Cref{subsec:prelimFPRAS,subsec:proofFPRAS} is the parity functions. For all positive integer $d$ define $\mathbf{Even}_{d}:\{0,1\}^{d}\rightarrow \{0,1\}$ by setting $\mathbf{Even}_{d}(x_{1},\dots,x_{d})=1$ if and only if $x_{1}+\dots+x_{d}$ is even.

\section{\#P-Hardness}\label{sec:hardness}
Let's define the range of parameters

$$\Gamma=\{(\beta,\gamma)\in\mathbb{R}^{2}:(\beta>\gamma)\wedge(-2<\beta+\gamma<1)\wedge(\gamma<0)\}\setminus\{(1,-1)\},$$ which will appear many times in this section.
Note that for $(\beta,\gamma) \in \Gamma$,
$\beta\gamma<(-2-\gamma)\gamma\leq 1$.

\subsection{Realizing Arbitrary Ratios}

The starting point for proving the hardness result \Cref{thm:sharpPhard} is to show that the ratio $R_{G,v}$ can take value in a dense subset of $\mathbb{R}$.

\begin{definition}
Given parameters $\beta,\gamma\in\mathbb{R}$, we say that a real number $r$ is realizable if there is a finite graph $G$ and a vertex $v\in V(G)$ such that $\left[Z_{G,v}\right]_{0}\neq 0$ and $R_{G,v}=r$. When the vertex $v$ is clear from context, we often just say $G$ realizes the number $r$.
\end{definition}
\begin{lemma}\label{lem:productrealize}
If $r_{1},r_{2}\in\mathbb{R}$ are realizable under parameters $\beta$ and $\gamma$, then $r_{1}r_{2}$ is also realizable.
\end{lemma}
\begin{proof}
If $R_{G_{1},v_{1}}=r_{1}$ and $R_{G_{2},v_{2}}=r_{2}$, take $G$ to be the ``wedge sum'' of $G_{1}$ and $G_{2}$, by first taking their disjoint union and then identifying $v_{1}$ and $v_{2}$ as a single vertex $v$. Then 
$$R_{G,v}=\frac{\left[Z_{G,v}\right]_{1}}{\left[Z_{G,v}\right]_{0}}=\frac{\left[Z_{G_{1},v_{1}}\right]_{1}\cdot \left[Z_{G_{2},v_{2}}\right]_{1}}{\left[Z_{G_{1},v_{1}}\right]_{0}\cdot \left[Z_{G_{2},v_{2}}\right]_{0}}=R_{G_{1},v_{1}}\cdot R_{G_{2},v_{2}}=r_{1}r_{2},$$
hence $r_{1}r_{2}$ is realizable.
\end{proof}
\begin{lemma}\label{lem:mobiusmap}
If $r\in\mathbb{R}$ is realizable and $r\neq -\beta$, then $\frac{1+\gamma r}{\beta + r}$ is also realizable.
\end{lemma}
\begin{proof}
If $R_{G,v}=r$, define a graph $G'$ with vertex set $V(G)\cup\{u\}$ and edge set $E(G)\cup\{\{u,v\}\}$, i.e. we attach a new edge to the vertex $v$ in $G$. Then
$$R_{G',u}=\frac{\left[Z_{G',u}\right]_{1}}{\left[Z_{G',u}\right]_{0}}=
\frac{\left[Z_{G,u,v}\right]_{1,0}+\left[Z_{G,u,v}\right]_{1,1}}{\left[Z_{G,u,v}\right]_{0,0}+\left[Z_{G,u,v}\right]_{0,1}}=
\frac{\left[Z_{G,v}\right]_{0}+\gamma\cdot \left[Z_{G,v}\right]_{1}}
{\beta\cdot \left[Z_{G,v}\right]_{0}+\left[Z_{G,v}\right]_{1}}=\frac{1+\gamma r}{\beta +r},
$$
hence $\frac{1+\gamma r}{\beta + r}$ is realizable.
\end{proof}
\begin{lemma}\label{lem:greaterthan1}
Let $\beta,\gamma$ be real numbers such that $(\beta,\gamma)\in\Gamma$. Then some real number in $(1,+\infty)$ is realizable.
\end{lemma}
\begin{proof}
We divide the proof into the following cases:\\
Case 1: $\beta+\gamma <0$ and $\beta\neq 0$. Take $V(G)=\{v\}$ and let $E(G)$ consist of 2 self loops on $v$. Then 
$$R_{G,v}=\left(\frac{\gamma}{\beta}\right)^{2}=1+\frac{(\beta-\gamma)(-\beta-\gamma)}{\beta^{2}}>1.$$
Case 2: $\beta=0$. Take $V(G)=\{v_{1},v_{2}\}$ and $E(G)=\{\{v_{1},v_{2}\},\{v_{2},v_{2}\}\}$. We have
$R_{G,v_{1}}=\frac{\beta+\gamma^{2}}{\beta^{2}+\gamma}=\gamma$.
By applying \Cref{lem:mobiusmap} 
since $\gamma \neq -\beta$, $(1+\gamma^2)/(\beta+\gamma)$ is realizable. By
applying~\Cref{lem:productrealize}, it follows that the real number 
$(1+\gamma^2)^2/\gamma^2$ is realizable. Since $-2<\gamma<0$ this quantity is at least~$4$. \\
Case 3: $\beta+\gamma\geq 0$ and $\gamma\neq -\beta^{2}$. Take $V(G)=\{v_{1},v_{2},v_{3}\}$ and $E(G)=\{\{v_{1},v_{2}\},\{v_{1},v_{3}\},\{v_{2},v_{2}\},\allowbreak
\{v_{3},v_{3}\}\}$. Then
$$R_{G,v_{1}}=\left(\frac{\beta+\gamma^{2}}{\beta^{2}+\gamma}\right)^{2}=1+\frac{(\beta-\gamma)(1-\beta-\gamma)(\beta^{2}+\gamma^{2}+\beta+\gamma)}{(\beta^{2}+\gamma)^{2}}>1.$$
Case 4: $\beta+\gamma\geq 0$ and $\gamma=-\beta^{2}$. Since $(\beta,\gamma)\neq  (1,-1)$, it follows that $-1<\gamma<0$ and $0<\beta+\gamma<1$. Take $V(G)=\{v_{1},v_{2},v_{3}\}$ and $E(G)=\{\{v_{1},v_{2}\},\{v_{2},v_{3}\},\{v_{3},v_{3}\}\}$. We have
$$R_{G,v_{1}}=\frac{\beta^{2}+\gamma+\gamma(\beta+\gamma^{2})}{\beta(\beta^{2}+\gamma)+\beta+\gamma^{2}}=\gamma.$$
 
By applying \Cref{lem:mobiusmap} 
since $\gamma \neq -\beta$, $(1+\gamma^2)/(\beta+\gamma)>1$ is realizable.
\end{proof}

\begin{lemma}\label{lem:minus1to0}
Let $\beta,\gamma$ be real numbers such that $(\beta,\gamma)\in\Gamma$. Then some real number in $(-1,0)$ is realizable.
\end{lemma}
\begin{proof}
Notice that in this range $-2<\beta+\gamma<2\beta$ and so $\beta>-1$. Consider the following 2 cases:\\
Case 1: $\gamma<-1$. Let $G$ consists of a single edge $\{v_{1},v_{2}\}$. We have 
$$R_{G,v_{1}}=\frac{1+\gamma}{\beta+1}=-1+\frac{\beta+\gamma+2}{\beta +1}>-1,$$
and, since $\gamma<-1$, $R_{G,v_1}$ is also less than 0. So $R_{G,v_{1}}$ gives a realizable ratio in $(-1,0)$.\\
Case 2: $\gamma\geq -1$. From \Cref{lem:greaterthan1}, and since we can take arbitrary powers due to \Cref{lem:productrealize}, we know some real number $r>-\frac{1}{\gamma}$ is realizable. 
Moreover, since $\beta\gamma<  1$, 
we have $-\beta < -1/\gamma$, so $r\neq -\beta$ and we can apply Lemma~\ref{lem:mobiusmap}. Also $\beta+r > \beta -1/\gamma > 0$. Appying the lemma,
we have
$$\frac{1+\gamma r}{\beta + r}=\gamma+\frac{1-\beta\gamma}{\beta + r}>\gamma\geq -1.$$
Since $r>-\frac{1}{\gamma}$, 
the quantity $(1+\gamma r)/(\beta+r)$
is also less than 0. So $\frac{1+\gamma r}{\beta + r}$ is a realizable ratio in $(-1,0)$. 
\end{proof}

\begin{proposition}\label{prop:dense}
Fix real parameters $\beta,\gamma$ such that $(\beta,\gamma)\in\Gamma$. For any real numbers $R\neq 0$ and $\varepsilon>0$, some real number strictly between $e^{-\varepsilon}R$ and $e^{\varepsilon}R$ is realizable.
\end{proposition}
\begin{proof}
We first assume that $R>0$. Take a realizable ratio $r_{0}$ greater than 1 (which exists by Lemma~\ref{lem:greaterthan1}), raise to the $k$th power (applying Lemma~\ref{lem:productrealize}) for some $k$ that is sufficiently large  that $r_0^k > -\beta$, and then apply \Cref{lem:mobiusmap}. This realizes a ratio
$$r_{1}=\frac{1+\gamma r_{0}^{k}}{\beta+r_{0}^{k}}=\gamma+\frac{1-\beta\gamma}{\beta +r_{0}^{k}}.$$
Since $\beta \gamma < 1$,
for a sufficiently large $k$, the ratio $r_{1}$ lies in the interval $(\gamma,e^{-\varepsilon}\gamma)$. Let
$$r_{2}=\frac{1+\gamma r_{0}^{k+1}}{\beta + r_{0}^{k+1}},$$
which  
satisfies $\gamma < r_2 < r_1 < \gamma e^{-\epsilon}$.
Thus, $r_{2}/r_{1}\in (1,e^{\varepsilon})$.

By Lemma~\ref{lem:minus1to0}, a number $r\in (-1,0)$ can be realized.
By Lemma~\ref{lem:productrealize}, the number $r r_0^j$ can be realized for any positive integer~$j$.
We will take $j$ large enough that $|r r_0^j | > 1/|r_1| > 1/|r_2|$. 
By Lemma~\ref{lem:productrealize}, the quantities 
$R_1 = r_1 r r_0^j$ and $R_2 = r_2 r r_0^j$ can be realized. These are in the range    $(1,+\infty)$  and have $R_2/R_1 = r_{2}/r_{1}\in (1,e^{\varepsilon})$. Moreover, 
by Lemma~\ref{lem:productrealize}, the quantity $R_3 = r^2 \in (0,1)$ can be realized.  To finish we will show that the multiplicative semigroup generated by $\{R_{1},R_{2},R_{3}\}$   intersects $(e^{-\varepsilon}R,e^{\varepsilon}R)$. To see this, consider the following system of inequalities:
\begin{equation}\label{ineq:system}
 \begin{cases}
R_{3}^{m}R_{1}^{n}\leq e^{-\varepsilon}R\\
R_{3}^{m}R_{2}^{n}\geq e^{\varepsilon}R.
\end{cases}   
\end{equation}

If some positive integers $m,n$ satisfy the above system of inequalities, then due to the fact that $R_{2}/R_{1}\in(1,e^{\varepsilon})$, at least one term of the geometric progression 
$$R_{3}^{m}R_{1}^{n},\quad R_{3}^{m}R_{1}^{n-1}R_{2},\quad\dots,\quad R_{3}^{m}R_{1}R_{2}^{n-1}, \quad R_{3}^{m}R_{2}^{n}$$
falls in the interval $(e^{-\varepsilon}R,e^{\varepsilon}R)$. What's more, as a product of realizable numbers, each term is realizable under the parameters $(\beta,\gamma)$. So it only remains to show that the system \eqref{ineq:system} has a positive integer solution. 

In order to ensure a solution for $n$, the requirements on $m$ are 
$$
R_{3}^{m} R_1<e^{-\varepsilon}R
$$
(this ensures a positive solution for $n$) and
$$
\log_{R_{1}}\left(\frac{e^{-\varepsilon}R}{R_{3}^{m}}\right)-\log _{R_{2}}\left(\frac{e^{\varepsilon}R}{R_{3}^{m}}\right)\geq 1
$$
(this ensures an integer solution for $n$). Using $R_2>R_1$, the latter simplifies to
\begin{equation}\label{eq:simplifiedlong}
m\ln\frac{1}{R_{3}}\geq \frac{\ln R_{1}\cdot\ln R_{2}+\varepsilon(\ln R_{1}+\ln R_{2})}{\ln R_{2}-\ln R_{1}}-\ln R.
\end{equation}
Since $R_{3}<1$, a sufficiently large integer $m$ satisfies both requirements. This concludes the proof in the case $R>0$. 

In the case $R<0$, pick any negative realizable ratio $r$, as in \Cref{lem:minus1to0}. Since $R/r>0$, we already know some real number in $(e^{-\varepsilon}R/r,e^{\varepsilon}R/r)$ is realizable. Multiplying it by $r$ gives a realizable ratio in $(e^{\varepsilon}R,e^{-\varepsilon}R)$.
\end{proof}
\begin{remark}
In \Cref{prop:dense}, we showed that the set of realizable ratios is dense in $\mathbb{R}$. However, we didn't control the size of the graph used in the approximation. It's worth noting that the dependency of the size on the accuracy parameter $\varepsilon$ is at least inverse linear: since $R_{2}/R_{1}=r_{2}/r_{1}=1+O(\varepsilon)$, i.e. $\log R_{2}-\log R_{1}=O(\varepsilon)$, by requirement \eqref{eq:simplifiedlong} the integer $m$ must be $\Omega(\varepsilon^{-1})$. This turns out to be insufficient on its own for proving \#P-hardness. In the following section, we will strengthen the dependency on $\varepsilon$ to polylogarithmic. In other words, we will approximately realize any ratio $R$ with exponential accuracy.
\end{remark}
\subsection{Exponential Accuracy}\label{subsec:exponentialaccuracy}
Actually, in addition to realizing with exponential accuracy, we must also \textit{efficiently compute} the graph $G$ that realizes a given ratio. This means it's necessary to quantify the computational expense. Assume $\beta$ and $\gamma$ are fixed real numbers, and that the input parameters $R$ and $\varepsilon$ are both rational numbers written in standard fraction forms. By ``polynomial-time algorithm'' we mean the running time is polynomial in the number of bits in the representations of the input parameters. In particular, since $\varepsilon$ is representable using $\log(\varepsilon^{-1})$ bits, the running time is polynomial in $\log(\varepsilon^{-1})$.

\begin{theorem}\label{thm:exponentialaccuracy}
Fix rational numbers $\beta,\gamma$ such that $(\beta,\gamma)\in\Gamma$. There is a polynomial-time algorithm that, given as input rational numbers $R>0$ and $\varepsilon>0$, outputs a graph $G$ and a vertex $v\in V(G)$ such that $\dfrac{[Z_{G,v}]_{1}}{[Z_{G,v}]_{0}}\in(e^{-\varepsilon}R,e^{\varepsilon}R)$.
\end{theorem}
\begin{proof} The proof is somewhat lengthy, so we divide it into several parts: 

\vspace{6pt}
\noindent \textbf{Part I: Preparations.} Our graph $G$ will have a path as its backbone, with additional gadgets attached on nodes:

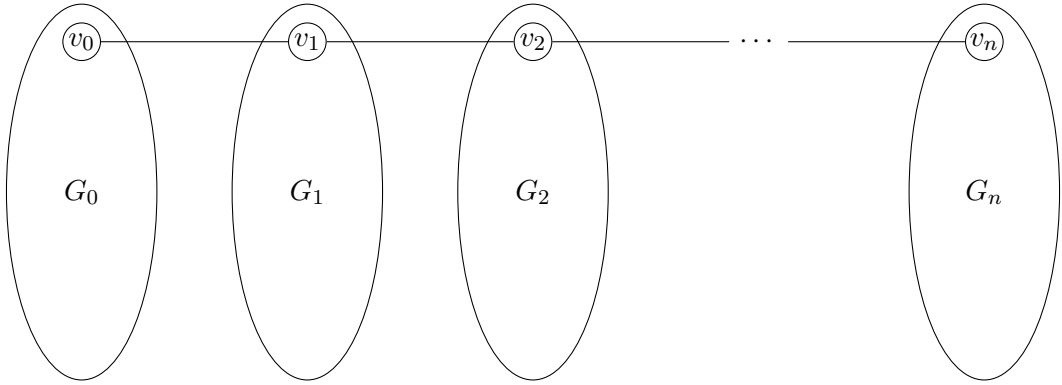
\begin{figure}[H]
\begin{center}
{
\tikzset{lab/.style={circle,draw,inner sep=0pt,fill=none,minimum size=5mm}}
\begin{tikzpicture}
\draw (0,3) node [lab] (0) {$v_{0}$};
\draw (3,3) node [lab] (1) {$v_{1}$};
\draw (6,3) node [lab] (2) {$v_{2}$};
\draw (9,3) node (3) {$\cdots$};
\draw (12,3) node [lab] (4) {$v_{n}$};
\draw (0,1) ellipse (1 and 2.5);
\draw (3,1) ellipse (1 and 2.5);
\draw (6,1) ellipse (1 and 2.5);
\draw (12,1) ellipse (1 and 2.5);
\draw (0,1) node {$G_{0}$};
\draw (3,1) node {$G_{1}$};
\draw (6,1) node {$G_{2}$};
\draw (12,1) node {$G_{n}$};
\draw (0)--(1)--(2)--(3)--(4);
\end{tikzpicture}
}
\end{center}
\caption{Structure of the Graph $G$} \label{fig:expaccgrpah}
\end{figure}
Formally, let $(v_{0},v_{1},\dots,v_{n})$ be a path, and let $G_{0}, G_{1},\dots, G_{n}$ be graphs realizing ratios $x_{0},x_{1},\dots, x_{n}$.

We form the graph $G$ by attaching $G_{0},\dots, G_{n}$ to the corresponding nodes on the path. It follows easily from \Cref{lem:mobiusmap} and \Cref{lem:productrealize} that if we denote the linear fractional transformation $r\mapsto \frac{1+\gamma r}{\beta +r}$ by $f$, the iteration 
\begin{equation}\label{eq:iteration}
y_{n}=x_{n}, \text{ and } y_{k}=x_{k}\cdot f(y_{k+1}),\text{ for }0\leq k\leq n-1
\end{equation}
gives $y_{0}=\dfrac{\left[Z_{G,v_{0}}\right]_{1}}{\left[Z_{G,v_{0}}\right]_{0}}$. So it suffices to compute the gadget graphs $G_{0},\dots, G_{n}$ such that the ratios $(x_{0},\dots,x_{n})$ they realize produce a $y_{0}\in(e^{-\varepsilon}R,e^{\varepsilon}R)$.

Before describing the algorithm, we need to prepare four ``landmarks'' $a,b,c,d$ on the real line, with $\gamma<a<b<0\leq |\beta|<c<d$. We require that $b=f(c)$, $a=f(d)$, and $\frac{b}{2d}<f'(c)$ (here $f'$ stands for the derivative of $f$). The existence of such rational numbers $a,b,c,d$ can be shown easily. For example, we can take $d=2c$ and let $c$ be sufficiently large. Observe that as $x\rightarrow +\infty$, the function value $f(x)$ approaches $\gamma$ from the right, and $f'(x)$ tends to 0 faster than $1/x$. So eventually $f'(c)$ gets closer to 0 than $b/2d$. 

We also prepare a gadget graph $H_{1}$ realizing a ratio $h_{1}\in\left(\sqrt{b/a},1\right)$, and a graph $H_{2}$ realizing a ratio $h_{2}\in (-\infty,-2)$. Their existence follows from \Cref{prop:dense}. 

\vspace{6pt}
\noindent \textbf{Part II: The Algorithm.} Until now, we have been describing information that doesn't depend on the input $(R,\varepsilon)$, and is thus hard-wired into our algorithm. Next we introduce the algorithm:

\vspace{6pt}

\begin{algorithm}[H]
\DontPrintSemicolon
\SetKwInOut{Input}{Input}\SetKwInOut{Output}{Output}
\caption{Realize Arbitrary Ratio}\label{alg:exponentialaccuracy}
\Input{$R,\varepsilon\in\mathbb{Q}^{>0}$}
\Output{A graph $G$ and a vertex $v_{0}$ of $G$ such that $\left[Z_{G,v_{0}}\right]_{1}/\left[Z_{G,v_{0}}\right]_{0}\in (e^{-\varepsilon}R,e^{\varepsilon}R)$.}
$k\gets 0$,
$R_{1}^{(0)}\gets R /(1+\varepsilon)  ,\quad R_{2}^{(0)}\gets R (1+\varepsilon)  $ \;
\While{$R_{2}^{(k)}/R_{1}^{(k)}<\sqrt{a/b}$}{
    \textbf{Compute} a graph $G_{k}$ realizing a ratio  $x_{k}\in\left(-R_{2}^{(k)}/\sqrt{ab},R_{2}^{(k)}/a\right)$
\tcp*{Using $H_{1}, H_{2}$}
    $R_{1}^{(k+1)}\gets f^{-1}(R_{1}^{(k)}/x_{k}),\quad R_{2}^{(k+1)}\gets f^{-1}(R_{2}^{(k)}/x_{k})$
\tcp*{$f^{-1}$ is the inverse of $f$}
    $k\gets k+1$\;
}
\textbf{Compute} a graph $G_{k}$ realizing a ratio $x_{k}\in(R_{1}^{(k)},R_{2}^{(k)})$
\tcp*{Using $H_{1}, H_{2}$}
$n\gets k$\;
Form a graph $G$ from $G_{0},\dots, G_{n}$ as in \Cref{fig:expaccgrpah}\;
\end{algorithm}

\vspace{6pt}
\noindent \textbf{Part III: Correctness.} To prove the correctness of \Cref{alg:exponentialaccuracy}, 
we first prove two claims about the numbers $x_{0},\dots,x_{n}$ and $R_{1}^{(0)},\dots, R_{1}^{(n)},R_{2}^{(0)},\dots,R_{2}^{(n)}$ computed in the course of the algorithm.

\begin{claim}\label{claim:betweencd}
$0<R_{1}^{(0)}<R_{2}^{(0)}$ and
for $k\in\{1,2,\dots,n\}$, $c<R_{1}^{(k)}<R_{2}^{(k)}<d$.
\end{claim}
\begin{proof}[Proof of Claim 1]\renewcommand{\qedsymbol}{$\blacksquare$}
We perform induction on $k$. The base case   $k=0$ is clear from the algorithm.   Now assume the claim holds for $k-1<n$. From the induction hypothesis and the choice of the ratio $x_{k-1}$, we have $x_{k-1}<R_{2}^{(k-1)}/a<0$ and hence $\gamma<R_{2}^{(k-1)}/x_{k-1}<R_{1}^{(k-1)}/x_{k-1}$. Since $f^{-1}$ is decreasing  on $(\gamma,+\infty)$,
we have
$$R_{2}^{(k)}=f^{-1}\left(\frac{R_{2}^{(k-1)}}{x_{k-1}}\right)<f^{-1}\left(\frac{R_{2}^{(k-1)}}{R_{2}^{(k-1)}/a}\right)=f^{-1}(a)=d,$$
and by the while loop condition
$$R_{1}^{(k)}=f^{-1}\left(\frac{R_{1}^{(k-1)}}{x_{k-1}}\right)>f^{-1}\left(\frac{R_{2}^{(k-1)}/\sqrt{\frac{a}{b}}}{x_{k-1}}\right)>f^{-1}\left(\frac{R_{2}^{(k-1)}/\sqrt{\frac{a}{b}}}{-R_{2}^{(k-1)}/\sqrt{ab}}\right)=f^{-1}(b)=c.$$
From the induction hypothesis we also clearly have 
\[R_{2}^{(k)}=f^{-1}\left(\frac{R_{2}^{(k-1)}}{x_{k-1}}\right)>f^{-1}\left(\frac{R_{1}^{(k-1)}}{x_{k-1}}\right)=R_{1}^{(k)}.\qedhere\]
\end{proof}
\begin{claim}\label{claimdone}
Running the iteration \eqref{eq:iteration} on the ratios $(x_{0},\dots,x_{n})$ produces a $y_{0}\in(e^{-\varepsilon}R,e^{\varepsilon}R)$.
\end{claim}
\begin{proof}[Proof of Claim \ref{claimdone}]\renewcommand{\qedsymbol}{$\blacksquare$}
Let the sequence $\{y_{k}\}_{0\leq k\leq n}$ be as in the iteration \eqref{eq:iteration}. We inductively show that $R_{1}^{(k)}< y_{k}<R_{2}^{(k)}$ holds for all $0\leq k\leq n$. For $k=n$, this is clear from the last \textbf{Compute} operation. Now assume $k\in\{0,1,\dots,n-1\}$ and the induction hypothesis holds for $k+1$. Combining with \Cref{claim:betweencd}, we know that $c<R_{1}^{(k+1)}<y_{k+1}<R_{2}^{(k+1)}<d$. Recall that $c>|\beta|$, and this ensures that $f$ is decreasing on $[c,+\infty)$. What's more, we know $x_{k}<0$ from the proof of \Cref{claim:betweencd}. So 
$$x_{k}\cdot f\left(R_{1}^{(k+1)}\right)<x_{k}\cdot f(y_{k+1})<x_{k}\cdot f\left(R_{2}^{(k+1)}\right),$$
i.e. $R_{1}^{(k)}<y_{k}<R_{2}^{(k)}$. Now, taking $k=0$ in the inductive hypothesis yields the claim.
\end{proof}

The correctness of \Cref{alg:exponentialaccuracy} then follows directly from \Cref{claimdone} (see \textbf{Part I} of this proof).

\vspace{6pt}
\noindent \textbf{Part IV: Efficiency.} It remains to show the efficiency of our algorithm. The next claim serves to bound the number of \textbf{Compute} operations executed:

\begin{claim}\label{claim:square}
For each $k\in\{1,2,\dots,n\}$, we have $R_{2}^{(k)}/R_{1}^{(k)}>\left(R_{2}^{(k-1)}/R_{1}^{(k-1)}\right)^{2}$.
\end{claim}
\begin{proof}[Proof of Claim 3]\renewcommand{\qedsymbol}{$\blacksquare$}
 We know from \Cref{alg:exponentialaccuracy} that 
 $R_{i}^{(k-1)}=x_{k-1}\cdot f\left(R_{i}^{(k)}\right)$, for $i\in\{1,2\}$. So
\begin{equation}\label{eq:claim3-1}
\begin{split}
\ln R_{2}^{(k-1)}-\ln R_{1}^{(k-1)}
 &= \ln\left(x_{k-1}\cdot f(R_{2}^{(k)})\right)-\ln\left( x_{k-1}\cdot f(R_{1}^{(k-1)})\right)\\
 &=\ln\left(-f(R_{2}^{(k)})\right)-\ln\left(-f(R_{1}^{(k)})\right).
\end{split}
\end{equation}
On the interval $[c,d]$, the function $x\mapsto \ln(-f(x))$ has derivative
$$\frac{f'(x)}{f(x)}=\frac{\beta-1/\gamma}{(x+\beta)(x+1/\gamma)}\leq \frac{\beta-1/\gamma}{(c+\beta)(c+1/\gamma)}=\frac{f'(c)}{f(c)},$$ where the inequality uses $c+1/\gamma>0$ (this is because the definition of~$f$ ensures that $c+1/\gamma=(\beta+c)f(c)\gamma^{-1}=(\beta+c)b\gamma^{-1}>0$).

Consequently, since $c < R_1^{(k)} < R_2^{(k)} < d$ by \Cref{claim:betweencd}, we have
\begin{equation}\label{eq:claim3-2}
\ln\left(-f(R_{2}^{(k)})\right)-\ln\left(-f(R_{1}^{(k)})\right)\leq \frac{f'(c)}{f(c)}\left(R_{2}^{(k)}-R_{1}^{(k)}\right).
\end{equation}
On the other hand, the derivative of $x\mapsto \ln x$ is at least $1/d$ on the interval $[c,d]$. So by \Cref{claim:betweencd} again
\begin{equation}\label{eq:claim3-3}
\ln R_{2}^{(k)}-\ln R_{1}^{(k)} \geq \frac{1}{d}\left(R_{2}^{(k)}-R_{1}^{(k)}\right).
\end{equation}
Combining \eqref{eq:claim3-1}, \eqref{eq:claim3-2} and \eqref{eq:claim3-3}, we have
$$\ln R_{2}^{(k-1)}-\ln R_{1}^{(k-1)}\leq \frac{f'(c)}{f(c)}\cdot d\cdot\left(\ln R_{2}^{(k)}-\ln R_{1}^{(k)}\right)
 < \frac{1}{2}\left(\ln R_{2}^{(k)}-\ln R_{1}^{(k)}\right),$$
where the final inequality uses the crucial property $b/(2d)<f'(c)$, i.e. $\frac{f'(c)}{f(c)}\cdot d<\frac{1}{2}$.
\end{proof}

From the while loop condition in \Cref{alg:exponentialaccuracy} we know that $R_{2}^{(n-1)}/R_{1}^{(n-1)}<\sqrt{a/b}$, while from the initialization of variables  $R_{2}^{(0)}/R_{1}^{(0)}= e^{2 \varepsilon}$. So it follows from \Cref{claim:square} that 
$ \sqrt{a/b} >  {R_2^{(n-1)}}/{R_1^{(n-1)}} > (e^{2\epsilon})^{2^{n-2}} $
so $n = O(\log(1/\epsilon))$.

The operation \textbf{Compute} is executed $(n+1)$ times in \Cref{alg:exponentialaccuracy}. Now that we have shown the number of \textbf{Compute} executions is polynomial in the input size, it suffices to show that each \textbf{Compute} operation takes polynomial time. 

The very first execution of \textbf{Compute} is a little bit different, where we need to realize a ratio in 
$\left( R/ ((1+\epsilon)\sqrt{ab}), (1+\epsilon) R /a\right)$.  Similar to the methods in \Cref{prop:dense}, let $s$ be a sufficiently large odd number such that $h_{2}^{s}<-   R/((1+\epsilon)\sqrt{ab})$, and then there must be an integer $t>0$ such that $h_{2}^{s}h_{1}^{t}$ falls into the interval, since $h_{1}\in\left(\sqrt{b/a},1\right)$. Both $s$ and $t$ are $O(|\log R|)$ and take polynomial time to compute. (The rational number $R$ is represented with at least $\Omega(|\log R|)$ bits.)

For $1\leq k< n$, we need to realize a ratio in $\left(-R_{2}^{(k)}/\sqrt{ab},R_{2}^{(k)}/a\right)$ in the    execution of \textbf{Compute} when the while loop is entered with value~$k$. By \Cref{claim:betweencd} the interval is contained in $(d/b,0)$. So the same method above applies and this time with a constant running time. 

Finally, we need to a realize a ratio in $(R_1^{(n)},R_2^{(n)})$.
By \Cref{claim:betweencd} we have $0 < c < R_1^{(n)} < R_2^{(n)} < d$, and the while loop condition in \Cref{alg:exponentialaccuracy} ensures that $R_2^{(n)}/R_1^{(n)}\geq \sqrt{a/b}$. Let $s$ be a sufficiently large even number such that $h_{2}^{s}>d$, and then there must be an integer $t>0$ such that $h_{2}^{s}h_{1}^{t}$ falls into the interval $\left(R_1^{(n)},R_2^{(n)}\right)$, since $h_{1}\in\left(\sqrt{b/a},1\right)$. Both $s$ and $t$ take constant time to compute.
\setcounter{claim}{0}
\end{proof}

For future reference, we want to be able to not only realize some ratio in a given interval, but also calculate exactly the ratio we realized. To avoid making the preceding theorem overly cumbersome, we state this as a separate proposition below. It is identical to \Cref{thm:exponentialaccuracy} except the addition of the final sentence, and that we also allow $R<0$. 
\begin{proposition}\label{prop:computeratio}
Fix rational numbers $\beta,\gamma$ such that $(\beta,\gamma)\in\Gamma$. There is a polynomial-time algorithm that, given as input rational numbers $R\neq 0$ and $\varepsilon>0$, outputs a graph $G$ and a vertex $v\in V(G)$ such that $\dfrac{[Z_{G,v}]_{1}}{[Z_{G,v}]_{0}}$ is strictly between $e^{-\varepsilon}R$ 
and $e^{\varepsilon}R$. The algorithm also outputs $[Z_{G,v}]_{1}$ and $[Z_{G,v}]_{0}$. 
\end{proposition}
\begin{proof}
We keep the notation in the proof of \Cref{thm:exponentialaccuracy} and give an additional procedure to calculate $\left[Z_{G,v_{0}}\right]$ on top of \Cref{alg:exponentialaccuracy}. 

First, we calculate the vector $\left[Z_{G_{k},v_{k}}\right]$ for each gadget graph $G_{k}$ attached to the path. Since each $G_{k}$ is a wedge sum of constant-sized gadget graphs, both $\left[Z_{G_{k},v_{k}}\right]_{0}$ and $\left[Z_{G_{k},v_{k}}\right]_{1}$ can be efficiently computed by multiplication (see \Cref{lem:productrealize}). 

We then compute $\left[Z_{G,v_{0}}\right]$ by the following recursive procedure:
$$B_{n}=\left[Z_{G_{n},v_{n}}\right], \text{ and }B_{k}=\left[Z_{G_{k},v_{k}}\right]\circ\left(\begin{bmatrix}
\beta & 1\\
1 & \gamma
\end{bmatrix}B_{k+1}\right),\text{ for }0\leq k\leq n-1,$$
where $\circ$ is the entry-wise product of two 2-by-1 vectors. It's easy to see that the result of the recursion, the vector $B_{0}$, is exactly $\left[Z_{G,v_{0}}\right]$. 

The case $R<0$ is easy to cope with by attaching to $v_{0}$ a gadget from \Cref{lem:minus1to0}. 
\end{proof}

\begin{remark}
To ensure that the algorithms in \Cref{thm:exponentialaccuracy} and \Cref{prop:computeratio} are polynomial-time in terms of bit complexity, we observe that all intermediate values computed are rational numbers representable with $\mathrm{poly}(R,1/\varepsilon)$ bits. This follows from the fact that the partition functions of   $\mathrm{poly}(R,1/\varepsilon)$-sized graphs are represented with $\mathrm{poly}(R,1/\varepsilon)$  bits. Throughout the remainder of \Cref{sec:hardness}, the bit complexity of all algorithms remains polynomial in the length of their inputs for the same reason; therefore, we omit explicit arguments about bit complexity going forward for brevity.
\end{remark}

\begin{remark}\label{remfive}
The results and proofs in this section are directly analogous to Proposition 15 of \cite{bezakova2019inapproximability}. Most importantly, our graph $G$ has the same path-iteration structure used in \cite{bezakova2019inapproximability}. The main difference between our proof and the one in \cite{bezakova2019inapproximability} is that, having no explicit ``contraction maps'' (see their Lemma 28) to rely on, our algorithm instead pivots on the landmarks $a,b,c,d$ and especially on the property $b/2d<f'(c)$, which helps achieve a similar contraction effect (see \Cref{claim:square} in the proof of our \Cref{thm:exponentialaccuracy}).
\end{remark}

\subsection{Simulating Ising Models}
In order to present the reduction for proving \Cref{thm:sharpPhard}, we need to be able to approximately realize a ferromagnetic Ising edge interaction using our interaction matrix $\begin{bmatrix} \beta & 1\\ 1 & \gamma\end{bmatrix}$. This is formulated in the next lemma:

\begin{lemma}\label{lem:realizeIsing}
Fix rational numbers $\beta,\gamma$ such that $(\beta,\gamma)\in\Gamma$. There is a polynomial-time algorithm that, given as input rational numbers $M^{*}>1$ and $\varepsilon>0$, outputs a graph $G$ and two vertices $u,v\in V(G)$ such that
$$[Z_{G,u,v}]=N\begin{bmatrix}
M_{0} & 1\\
1 & M_{1}
\end{bmatrix},$$
for some $N>0$ and $M_{0},M_{1}>M^{*}$ such that $M_{1}/M_{0}\in(1,e^{\varepsilon})$. The algorithm also outputs the exact matrix $[Z_{G,u,v}]$ it realized.
\end{lemma}
\begin{proof}
For technical reasons, we first assume $\beta+\gamma\neq 0$. Let $P$ be a length-2 path with endpoints $u,v$. Formally, let its vertex set be $\{u,w,v\}$ and its edge set be $\{\{u,w\},\{w,v\}\}$. It's easy to calculate 
$$[Z_{P,u,v}]=\begin{bmatrix}
\beta^{2}+1 & \beta+\gamma\\
\beta+\gamma & \gamma^{2}+1
\end{bmatrix}.$$
Now let $G$ be a parallel connection of $2k$ paths like $P$, with gadget graphs $G_{1}, G_{2}$ attached to the two end points:

\begin{center}
{
\tikzset{lab/.style={circle,draw,inner sep=0pt,fill=none,minimum size=5mm}}
\begin{tikzpicture}
\draw (-2,0) node [lab] (0) {$u$};
\draw (2,0) node [lab] (1) {$v$};
\draw (0,2) node [lab] (2) {};
\draw (0,1) node [lab] (3) {};
\draw (0,0) node (4) {$\vdots$};
\draw (0,-1) node [lab] (5) {};
\draw (0,-2) node [lab] (6) {};
\draw (0)--(2)--(1);
\draw (0)--(3)--(1);
\draw (0)--(5)--(1);
\draw (0)--(6)--(1);
\draw (-3,0) ellipse (1.5 and 1);
\draw (-3,0) node {$G_{1}$};
\draw (3,0) ellipse (1.5 and 1);
\draw (3,0) node {$G_{2}$};
\end{tikzpicture}
}
\end{center}

The algorithm consists of the following steps:
\begin{itemize}
\setlength\itemsep{0pt}
\item Compute the smallest integer $k$ such that $M\triangleq\left(\dfrac{(\beta^{2}+1)(\gamma^{2}+1)}{(\beta+\gamma)^{2}}\right)^{k}> e^{\varepsilon}M^{*}$. Note that $$\dfrac{(\beta^{2}+1)(\gamma^{2}+1)}{(\beta+\gamma)^{2}}=1+\dfrac{(1-\beta\gamma)^{2}}{(\beta+\gamma)^{2}}>1.$$
\item Compute a graph $G_{1}$ that realizes a ratio $R$ with
$$\left(\frac{\beta^{2}+1}{\gamma^{2}+1}\right)^{k}<R<e^{\varepsilon/2}\left(\frac{\beta^{2}+1}{\gamma^{2}+1}\right)^{k},$$ 
using the procedure in \Cref{thm:exponentialaccuracy}. Attach $G_{1}$ to the vertex $u$, and attach an isomorphic copy $G_{2}$ to the vertex $v$.   
\end{itemize}

Since $k=O(\log M^{*})$, the first step clearly runs in polynomial time. The guarantee of \Cref{thm:exponentialaccuracy} tells us that the second step runs in $\mathrm{poly}(k,\log\frac{1}{\varepsilon})$ time, which is again polynomial in the size of our inputs since $k=O(\log M^{*})$. So we have verified the efficiency of the algorithm. 

As to the correctness, it suffices to observe that 
\begin{align}
\label{eq:computematrix1}
[Z_{G,u,v}]_{0,0}&=(\beta^{2}+1)^{2k}\Big(\left[Z_{G_{1,u}}\right]_{0}\Big)^{2},\\
\label{eq:computematrix2}
[Z_{G,u,v}]_{0,1}=[Z_{G,u,v}]_{1,0}&=(\beta+\gamma)^{2k}\Big(\left[Z_{G_{1,u}}\right]_{0}\Big)^{2}R\triangleq N,\\
\label{eq:computematrix3}
[Z_{G,u,v}]_{1,1}&=(\gamma^{2}+1)^{2k}\Big(\left[Z_{G_{1,u}}\right]_{0}\Big)^{2}R^{2}.
\end{align}
It follows that
\begin{itemize}
\setlength\itemsep{0pt}
\item $M_{0}\triangleq\dfrac{[Z_{G,u,v}]_{0,0}}{N}=\dfrac{(\beta^{2}+1)^{2k}}{(\beta+\gamma)^{2k}R}>\dfrac{(\beta^{2}+1)^{k}(\gamma^{2}+1)^{k}}{e^{\varepsilon/2}(\beta+\gamma)^{2k}}=Me^{-\varepsilon/2}>M^{*}$, and
\item $\dfrac{M_{1}}{M_{0}}\triangleq\dfrac{[Z_{G,u,v}]_{1,1}}{[Z_{G,u,v}]_{0,0}}=\dfrac{(\gamma^{2}+1)^{2k}R^{2}}{(\beta^2+1)^{2k}}\in(1,e^{\varepsilon})$.
\end{itemize}

Finally, given equations \eqref{eq:computematrix1}, \eqref{eq:computematrix2} and \eqref{eq:computematrix3}, and \Cref{prop:computeratio}, we can easily compute the 2-by-2 matrix $[Z_{G,u,v}]$ exactly. 

This concludes the proof of the lemma in the case $\beta+\gamma\neq 0$. The case $\beta+\gamma=0$ can be solved with a small tweak of parameters. In fact, we can perturb the parameters using gadgets given by \Cref{prop:dense}:
\begin{center}
{
\tikzset{lab/.style={circle,draw,inner sep=0pt,fill=none,minimum size=5mm}}
\begin{tikzpicture}
\draw (-4,0) node [lab] (0) {};
\draw (-2,0) node [lab] (1) {};
\draw (0)--(1);
\draw (0,0) node {$\Longrightarrow$};
\draw (4,0) node [lab] (2) {};
\draw (6,0) node [lab] (3) {};
\draw (2)--(3);
\draw (3,0) ellipse (1.5 and 1);
\draw (7,0) ellipse (1.5 and 1);
\draw (3,0) node {$H_{1}$};
\draw (7,0) node {$H_{2}$};
\end{tikzpicture}
}
\end{center}
For each edge with an interaction matrix $\begin{bmatrix}
\beta & 1\\
1 & \gamma
\end{bmatrix}$, by attaching a gadget graph realizing a ratio $r$ to each of its endpoints, we can turn the interaction matrix into $\begin{bmatrix}\beta r^{-1} & 1\\ 1 & \gamma r
\end{bmatrix}$, up to a normalization factor. For any pair $(\beta,\gamma)$ with $\beta+\gamma=0$ and in the range $\Gamma$, one can always use \Cref{prop:dense} to prepare a gadget graph that realizes a ratio $r\in(1,1+\frac{1}{\beta})$ so that the perturbed parameter pair $(\beta r^{-1},\gamma r)$ still lies in the range $\Gamma$ but is no longer on the line $\{x+y=0\}$. So we have reduced the case $\beta+\gamma=0$ to the case $\beta+\gamma\neq 0$, which is already solved.
\end{proof}
\subsection[Hardness]{Proof of \Cref{thm:sharpPhard}}

Having the crucial \Cref{prop:computeratio} and \Cref{lem:realizeIsing} in place, we are finally ready to prove \Cref{thm:sharpPhard}. The proof follows the approach of \cite{goldberg2014complexity}, reducing from the following problem:

\begin{itemize}[leftmargin=0pt]
\setlength\itemsep{0pt}
  \item[] \textbf{Name} {\sc \#Minimum Cardinality $(s,t)$-Cut}.
  \item[] \textbf{Instance} A graph $G=(V,E)$ and distinguished vertices~$s,t\in V$.
  \item[] \textbf{Output} $|\{S\subseteq E:\mbox{$S$ is a minimum cardinality $(s,t)$-cut in $G$}\}|$.
\end{itemize}

\begin{proof}[Proof of \Cref{thm:sharpPhard}]

If $\beta=\gamma$, then from $\beta+\gamma>-2$ and $\min\{\beta,\gamma\}<0$ it follows that $\beta=\gamma\in(-1,0)$. The theorem then follows from \Cref{prop:isingsharpP}. So we may assume that $\beta\neq \gamma$. By symmetry between $\beta$ and $\gamma$, it's then without loss of generality to assume $\gamma<\beta$. This places us in the range $\Gamma$, and hence in particular, \Cref{prop:computeratio} and \Cref{lem:realizeIsing} apply.  

We give a Turing reduction from {\sc \#Minimum Cardinality $(s,t)$-Cut}, which was shown to be \#P-hard by \cite{provan1983complexity}, to the problem of determining the sign of the partition function $Z_{G}$. 

Let $(G,s,t)$ be an instance of {\sc \#Minimum Cardinality $(s,t)$-Cut}. Assume without loss of generality that $G$ is connected. Let $n=|V(G)|$ and $m=|E(G)|$. Let $k$ be the size of a minimum cardinality $(s,t)$-cut in $G$, and let $C$ be the number of size-$k$ $(s,t)$-cuts, both of which are unknown. In order to compute $C$, we will create a sequence of graphs based on $G$, and feed them into the oracle that computes the sign of the partition function.  

First, we run the procedure in \Cref{lem:realizeIsing}, on input $M^{*}=2^{5m}$ and $\varepsilon=2^{-4m}$. This gives us a gadget graph $H$, two distinguished terminals among its vertices, and rational numbers $N,M_{0},M_{1}>0$, such that 
\begin{enumerate}[label=(\alph*),itemsep=0pt]
\item $2^{5m}<M_{0}<M_{1}<e^{2^{-4m}}M_{0}$, and
\item The graph $H$ realizes an interaction matrix
$N\begin{bmatrix}
M_{0} & 1\\
1 & M_{1}
\end{bmatrix}$ between its two terminals.
\end{enumerate}

Create a graph $G'$ by replacing every edge $\{u,v\}\in E(G)$ with a copy of the gadget graph $H$. We claim that the 2-by-2 matrix $\left[Z_{G',s,t}\right]$ contains very accurate information about $C$. For example, assume that in a spin configuration of $G$, the source $s$ is fixed to have spin $0$ and the sink $t$ to have spin $1$. Let $\Omega=\{\sigma\in\{0,1\}^{V(G)}:\sigma(s)=0,\sigma(t)=1\}$ be all possible spin configurations conditional on the spins of $s$ and $t$. Then a minimum-cardinality $(s,t)$-cut is equivalent to a configuration $\omega\in\Omega$ that minimizes the number of edges with differing spins on the endpoints. This set of configurations corresponding to minimum-cardinality $(s,t)$-cuts are denoted by $\Omega_{0}$, which has exactly $C$ elements. We then have 
\begin{align*}
\frac{\left[Z_{G',s,t}\right]_{1,0}}{N^{m}}&=\frac{1}{N^{m}}\sum_{\sigma\in\Omega}\prod_{\{u,v\}\in E}A_{\sigma(u),\sigma(v)} &\left(\text{where }A=N\begin{bmatrix}
M_{0} & 1\\
1 & M_{1}
\end{bmatrix}\right)\\
&=\sum_{\sigma\in \Omega_{0}}\prod_{\{u,v\}\in E}\frac{A_{\sigma(u),\sigma(v)}}{N}+\sum_{\sigma\in\Omega\setminus\Omega_{0}}\prod_{\{u,v\}\in E}\frac{A_{\sigma(u),\sigma(v)}}{N}\\
&\leq CM_{1}^{m-k}+2^{m}M_{1}^{m-k-1}\\
&\leq CM_{1}^{m-k}(1+2^{m}M_{1}^{-1})\\
&\leq (1+2^{-4m})CM_{1}^{m-k}.
\end{align*}

On the other hand, we have the obvious lower bound $\left[Z_{G',s,t}\right]_{1,0}/N^{m}\geq CM_{0}^{m-k}$.  Similarly we can obtain estimates for the other entries of $\left[Z_{G',s,t}\right]$: 
\begin{equation}\label{eq:sharopestimates}
\begin{split}
\left[Z_{G',s,t}\right]_{0,0}/N^{m},\left[Z_{G',s,t}\right]_{1,1}/N^{m}&\in\big(M_{0}^{m} ,(1+2^{-4m})M_{1}^{m}\big)\\
\left[Z_{G',s,t}\right]_{1,0}/N^{m},\left[Z_{G',s,t}\right]_{0,1}/N^{m}&\in\left(CM_{0}^{m-k}, (1+2^{-4m})CM_{1}^{m-k}\right)
\end{split}
\end{equation}

Since $M_{1}/M_{0}\in(1,e^{2^{-4m}})$, the lower bounds matches the upper bounds up to an exponentially small multiplicative constant. Thus, $C$ and $k$ are the crucial information determining the matrix $\left[Z_{G',s,t}\right]$. But in order to extract them exactly using the sign oracle, more work is needed.

Assume we can generate two gadget graphs $H_{1},H_{2}$ that realize vertex activity vectors $\begin{bmatrix}N_{1}\\ N_{1}h_{1}\end{bmatrix}$ and $\begin{bmatrix}N_{2}\\ N_{2}h_{2}\end{bmatrix}$, respectively, for some $N_1$, $h_1$, $N_2$ and $h_2$. By attaching them to the distinguished vertices $s$ and respectively $t$, we obtain a graph $G'_{H_{1},H_{2}}$ with a partition function 

$$Z_{G'_{H_{1},H_{2}}}=N_{1}N_{2}\Big(\left[Z_{G',s,t}\right]_{0,0}+h_{1}\left[Z_{G',s,t}\right]_{1,0}+h_{2}\left[Z_{G',s,t}\right]_{0,1}+h_{1}h_{2}\left[Z_{G',s,t}\right]_{1,1}\Big).$$

If the numbers $N_{1},N_{2},h_{1},h_{2}$ are known and are nonzero, as is the case if we have generated $H_{1}$ and $H_{2}$ using \Cref{prop:computeratio}, by feeding $G'_{H_{1},H_{2}}$ into the oracle for determining the sign of the partition function, we can determine the sign of the function 
$$T(x,y):=\frac{1}{1+xy}\frac{\left[Z_{G',s,t}\right]_{0,0}}{N^{m}}+\frac{x}{1+xy}\frac{\left[Z_{G',s,t}\right]_{1,0}}{N^{m}}+\frac{y}{1+xy}\frac{\left[Z_{G',s,t}\right]_{0,1}}{N^{m}}+\frac{xy}{1+xy}\frac{\left[Z_{G',s,t}\right]_{1,1}}{N^{m}}$$
at $x=h_{1}$, $y=h_{2}$, as
$$T(h_{1},h_{2})=\frac{1}{(1+h_{1}h_{2})N_{1}N_{2}N^{m}}\cdot Z_{G'_{H_{1},H_{2}}}.$$
Our claim is that, by trying suitably generated $H_{1}$'s and $H_{2}$'s, from the values of all the $\mathrm{sgn}(T(h_{1},h_{2}))$ we get, the number $C$ can be determined exactly. 

We proceed by a sandwiching argument. Define linear functions $L(x)=M_{0}^{m}+(1+2^{-4m})CM_{1}^{m-k}x$ and $U(x)=(1+2^{-4m})M_{1}^{m}+CM_{0}^{m-k}x$. Using the bounds \eqref{eq:sharopestimates}, it's easy to see that for all $x,y<0$
$$L\left(\frac{x+y}{1+xy}\right)<T(x,y)<U\left(\frac{x+y}{1+xy}\right).$$

In particular, if the oracle tells us $T(h_{1},h_{2})>0$ then we know that $U(\frac{h_{1}+h_{2}}{1+h_{1}h_{2}})>0$, and otherwise we know $L(\frac{h_{1}+h_{2}}{1+h_{1}h_{2}})< 0$. Combining this observation with a standard binary search procedure, we can approximately determine the zeros of $L$ and $U$, where the information about $C$ and $k$ actually lies.

\vspace{6pt}
\begin{algorithm}[H]
\DontPrintSemicolon
\SetKwInOut{Input}{Input}\SetKwInOut{Output}{Output}
\caption{Binary Search for Zero}\label{alg:binarysearch}
\Input{Oracle access to the function $(H_{1},H_{2})\mapsto \mathrm{sgn}(T(\frac{h_{1}+h_{2}}{1+h_{1}h_{2}}))$}
\Output{$p,q<0$ such that $L(q)<0$, $U(p)>0$, and $q/p<e^{2^{-4m}}$}
Initialize variables $p\gets -4, \quad q\gets -M_{1}^{m}$
\tcp*{Clearly $L(q)<0$ and $U(p)>0$ hold}
\While{$q/p\geq \exp(2^{-4m})$}{
    Pick rational numbers $r\in \left( -|p|^{4/9}|q|^{5/9},-|p|^{5/9}|q|^{4/9}\right)$ and $\varepsilon\leq(100 |r| )^{-1}2^{-4m}$\;
    
    Use \Cref{prop:computeratio} with inputs $R_1=-1/2$ and~$\epsilon$ to get a graph $H_{1}$ which realizes a ratio $h_{1}\in (e^{\epsilon} R_1, e^{-\epsilon} R_1)$\;

    Use \Cref{prop:computeratio} with inputs $R_2=\frac{2r+1}{2+r}$ and~$\varepsilon$ to get a graph $H_{2}$ which realizes a ratio $h_{2}\in (e^{-\epsilon} R_2, e^{\epsilon} R_2)$
    \tcp*{
    By \Cref{lem:auxitoalg2}, $\frac{h_{1}+h_{2}}{1+h_{1}h_{2}}\in\left(-|p|^{1/3}|q|^{2/3},-|p|^{2/3}|q|^{1/3}\right)$}
    Use the oracle to get the sign of $T(h_{1},h_{2})$\;
    \eIf{$T(h_{1},h_{2})>0$}{
        Then $U(\frac{h_{1}+h_{2}}{1+h_{1}h_{2}})>0$, and let $p\gets\frac{h_{1}+h_{2}}{1+h_{1}h_{2}}$ \tcp*{$L(q)<0$ and $U(p)>0$ still hold}
    }{
        Then $L(\frac{h_{1}+h_{2}}{1+h_{1}h_{2}})<0$, and let $q\gets\frac{h_{1}+h_{2}}{1+h_{1}h_{2}}$ \tcp*{$L(q)<0$ and $U(p)>0$ still hold}
    }
}
\end{algorithm}

 \vspace{6pt}

The reason \Cref{alg:binarysearch} runs in polynomial time is as follows:
\begin{itemize}[itemsep=0pt]
\item In each iteration of the while loop, since $\frac{h_{1}+h_{2}}{1+h_{1}h_{2}}\in\left(-|p|^{1/3}|q|^{2/3},-|p|^{2/3}|q|^{1/3}\right)$ by \Cref{lem:auxitoalg2}, at the end of the iteration $\ln(q/p)$ shrinks by at least a factor of $1/3$.
\item The initial value of $\ln(q/p)$ is no more than $m \ln(M_1)$, which is at most polynomial in~$m$ since the algorithm of Lemma~\ref{lem:realizeIsing} outputs~$M_1$ in polynomial time given input $M^* = 2^{5m}$.
\end{itemize}

The outcome of \Cref{alg:binarysearch} is a pair $p,q<0$ with $q/p\in(1,e^{2^{-4m}})$ such that $L(q)<0$ and $U(p)>0$. Now, we have\
\begin{equation}\label{eq:upperboundC}
U(p)>0\Rightarrow C<\frac{(1+2^{-4m})M_{1}^{m}}{(-p)M_{0}^{m-k}}
\end{equation}
and
\begin{equation}\label{eq:lowerboundC}
L(q)<0\Rightarrow C>\frac{M_{0}^{m}}{(1+2^{-4m})(-q)M_{1}^{m-k}}.
\end{equation}

The ratio of the upperbound for $C$ to the lowerbound is at most
$$(1+2^{-4m})^{2}\cdot\frac{q}{p}\left(\frac{M_{1}}{M_{0}}\right)^{2m-k}\leq \exp\left(2\cdot 2^{-4m}+2^{-4m}+2m\cdot 2^{-4m}\right)<\frac{2^{m}}{2^{m}-1}.$$
This means that, for any given $k$, there is at most one integer in $\{1,2,\dots,2^{m}\}$ that lies between the lower bound and the upper bound. Since $C$ must be in $\{1,2,\dots,2^{m}\}$, if $k$ is determined, we can obtain a unique solution for $C$. But on the other hand, the fact $M_{0},M_{1}>2^{4m}$ implies that there is at most one value of $k$ that gives a solution $C$ in the right range. Therefore, both $k$ and $C$ are efficiently computable using the bounds \eqref{eq:upperboundC} and $\eqref{eq:lowerboundC}$.
\end{proof}

During the above proof of \Cref{thm:sharpPhard}, we have made use of the following technical lemma:
 \begin{lemma}\label{lem:auxitoalg2}
Let $q<p\leq -4$ be real numbers. Suppose $-|p|^{4/9}|q|^{5/9}<r<-|p|^{5/9}|q|^{4/9}$ and $0<\varepsilon\leq (100 |r|)^{-1}\cdot\min\{\ln(q/p),1\}$. If
\[-\frac{1}{2}e^{\varepsilon}< h_{1}<-\frac{1}{2}e^{-\varepsilon}\quad\text{and}\quad\frac{2r+1}{2+r}e^{-\varepsilon}<h_{2}<\frac{2r+1}{2+r}e^{\varepsilon},\]
then $-|p|^{1/3}|q|^{2/3}<\frac{h_{1}+h_{2}}{1+h_{1}h_{2}}<-|p|^{2/3}|q|^{1/3}$.
\end{lemma} 
\begin{proof}
Note that if $\varepsilon=0$, then the assumption would force $h_{1}=-\frac{1}{2}$ and $h_{2}=\frac{2r+1}{2+r}$, in which case $\frac{h_{1}+h_{2}}{1+h_{1}h_{2}}$ exactly equals $r$. So it suffices to control the deviation from this ideal case in terms of the error parameter $\varepsilon>0$.

Since $r<p\leq -4$, we have $R:=\frac{2r+1}{2+r}>2$. As it is guaranteed that $\varepsilon\leq 1/100$, we can use the inequalities $e^{\varepsilon}\leq 1+2\varepsilon$, $e^{-\varepsilon}\geq 1-\varepsilon$ and $1+6\varepsilon\leq e^{6\varepsilon}\leq 1+12\varepsilon$ to obtain
\begin{align*}
h_{1}+h_{2}< Re^{\varepsilon}-\frac{1}{2}e^{-\varepsilon}
&= \left(R-\frac{1}{2}\right)e^{6\varepsilon}+R(e^{\varepsilon}-e^{6\varepsilon})+\frac{1}{2}(e^{6\varepsilon}-e^{-\varepsilon})\\
&\leq \left(R-\frac{1}{2}\right)e^{6\varepsilon}+2(2\varepsilon-6\varepsilon)+\frac{1}{2}(12\varepsilon+\varepsilon)\\
&\leq \left(R-\frac{1}{2}\right)e^{6\varepsilon}.
\end{align*}
Using the estimates $e^{\varepsilon}\leq 1+2\varepsilon$, $e^{-\varepsilon}\geq 1-\varepsilon$ and $1-6\varepsilon\leq e^{-6\varepsilon}\leq 1-3\varepsilon$, we have
\begin{align*}
h_{1}+h_{2}> Re^{-\varepsilon}-\frac{1}{2}e^{\varepsilon}
&= \left(R-\frac{1}{2}\right)e^{-6\varepsilon}+R(e^{-\varepsilon}-e^{-6\varepsilon})+\frac{1}{2}(e^{-6\varepsilon}-e^{\varepsilon})\\
&\geq \left(R-\frac{1}{2}\right)e^{-6\varepsilon}+2\left(-\varepsilon+3\varepsilon\right)+\frac{1}{2}(-6\varepsilon-2\varepsilon)\\
&= \left(R-\frac{1}{2}\right)e^{-6\varepsilon}.
\end{align*}

It's also guaranteed that $|r|\varepsilon<1/100$. Using the estimates $e^{-2\varepsilon}\geq 1-2\varepsilon$ and $e^{-4|r|\varepsilon}\leq 1-2|r|\varepsilon$, we have
\begin{align*}
-h_{1}h_{2}-1>\frac{R}{2}e^{-2\varepsilon}-1 &=\left(\frac{R}{2}-1\right)\left(1-\frac{R}{R-2}(1-e^{-2\varepsilon})\right)\\
&\geq \left(\frac{R}{2}-1\right)(1-|r|\cdot 2\varepsilon)\\
&\geq \left(\frac{R}{2}-1\right)e^{-4|r|\varepsilon}.
\end{align*}
Using the estimates $e^{2\varepsilon}\leq 1+4\varepsilon$ and $e^{4|r|\varepsilon}\geq 1+4|r|\varepsilon$ we have
\begin{align*}
-h_{1}h_{2}-1<\frac{R}{2}e^{2\varepsilon}-1 &=\left(\frac{R}{2}-1\right)\left(1+\frac{R}{R-2}(e^{2\varepsilon}-1)\right)\\
&\leq \left(\frac{R}{2}-1\right)(1+|r|\cdot 4\varepsilon )\\
&\leq \left(\frac{R}{2}-1\right)e^{4|r|\varepsilon }.
\end{align*}
In conclusion, we have
$$\frac{h_{1}+h_{2}}{1+h_{1}h_{2}}r^{-1}=\frac{h_{1}+h_{2}}{R-1/2}\left(\frac{h_{1}h_{2}+1}{1-R/2}\right)^{-1}< \exp(6\varepsilon+4|r|\varepsilon)\leq \left(\frac{q}{p}\right)^{1/9}$$
and
$$\frac{h_{1}+h_{2}}{1+h_{1}h_{2}} r^{-1}=\frac{h_{1}+h_{2}}{R-1/2}\left(\frac{h_{1}h_{2}+1}{1-R/2}\right)^{-1}> \exp(-6\varepsilon-4|r|\varepsilon)\geq \left(\frac{q}{p}\right)^{-1/9}.$$
Combining these with the assumption $-|p|^{4/9}|q|^{5/9}<r<-|p|^{5/9}|q|^{4/9}$, we conclude that 
\[-|p|^{1/3}|q|^{2/3}<\frac{h_{1}+h_{2}}{1+h_{1}h_{2}}<-|p|^{2/3}|q|^{1/3}.\qedhere\]
\end{proof}

 \section{Approximation Schemes}\label{sec:approx}
In this section, we give the two approximation schemes promised in \Cref{thm:FPTAS} and \Cref{thm:FPRAS}.  
\subsection{Preliminaries for the FPTAS}
The deterministic approximation scheme of \Cref{thm:FPTAS} will mainly rely on the powerful zero-freeness framework. In particular, our main tool is the following lemma developed and proved in \cite{barvinok2016combinatorics} and \cite{patel2017deterministic}.

\begin{lemma}\label{lem:zerofreeness}
Fix rational numbers $\beta$ and $\gamma$. Let $U$ be an open set in the complex plane that contains the real interval $[0,\lambda]$ for some $\lambda\in\mathbb{Q}^{+}$. Suppose that for all graphs $G$ the polynomial $Z_{G}(x)$ has no complex root in $U$. Then for any positive integer $\Delta$, there exists an FPTAS for $Z_{G}(\lambda)$, where $G$ is an input graph of maximum degree no more than $\Delta$ (without the bounded degree requirement, there is a quasi-polynomial time approximation scheme).
\end{lemma}

Our method for showing zero-freeness is the classical contraction method. It was first introduced in \cite{asano1970theorems} to give a simple proof for the Lee-Yang circle theorem \cite{lee1952statistical}, and was further extended in \cite{ruelle1971extension}. 
These results have been used previously in the area of algorithmic counting, e.g., by
Sinclair and Srivastava~\cite{SinclairSrivastava} and by Guo, Liao, Lu, and Zhang~\cite{GLLZ}. We will use the theorem of \cite{ruelle1971extension} in the following form.

\begin{lemma}\label{lem:contraction}
For each $\ell\in[m]$, let $K_\ell$ be a subset of the complex plane $\mathbb{C}$ that doesn't contain 0. Suppose the complex multi-affine polynomial
$$P(z_{1},\dots,z_{m})=\sum_{I\subseteq [m]}F(I)\prod_{\ell\in I}z_{\ell},$$
where each $F(I)$ is a complex coefficient, vanishes (i.e. attains value 0) only when $z_{\ell}\in K_{\ell}$ for some $\ell\in[m]$. Write $[m]$ as a disjoint union of subsets $I_{1},\dots,I_{n}$. Then the complex multi-affine polynomial
$$Q(w_{1},\dots,w_{n}):=\sum_{J\subseteq [n]}F\left(\bigcup_{j\in J}I_{j}\right)\prod_{j\in J}w_{j}$$
can vanish only when 
$w_{j}\in (-1)^{|I_{j}|+1}\prod_{\ell\in I_{j}}K_{\ell}$ for some $j\in [n]$, where the product is the Minkowski product of sets, meaning that
$\prod_{\ell\in I_j} K_\ell := \left\{ \prod_{\ell\in I_{j}}x_{\ell}\mid \forall \ell\in I_{j},\;x_{\ell}\in K_{\ell}\right\}$.
\end{lemma}

\begin{proof}[Proof Sketch]
The transformation from the polynomial $P$ to the polynomial $Q$ amounts to ``contracting'' subsets of variables into single variables. We may perform the contraction iteratively, where in each step we contract two variables into one. By induction it suffices to analyze each individual contraction step. For each contraction step, it suffices to prove the following statement: if a complex polynomial $a_{0}+a_{1}z_{\ell_{1}}+a_{2}z_{\ell_{2}}+a_{3}z_{\ell_{1}}z_{\ell_{2}}$ vanishes only when $z_{\ell_{1}}\in K_{\ell_{1}}$ or $z_{\ell_{2}}\in K_{\ell_{2}}$, then the polynomial $a_{0}+a_{3}w$ vanishes only when $w\in -K_{\ell_{1}}\cdot K_{\ell_{2}}$. We refer to \cite[Proof of Main Lemma]{ruelle1971extension} for a proof of this statement.
\end{proof}

The following corollary is all we need \Cref{lem:contraction} for:
\begin{corollary}\label{cor:contractunitcircle}
Fix real parameters $\beta$ and $\gamma$. Assume that the polynomial $\gamma z_{1}z_{2}+z_{1}+z_{2}+\beta$ doesn't vanish when $|z_{1}|,|z_{2}|<r$, for some $r>0$. Then for any graph $G$, the partition function $Z_{G}(\bflam)$ doesn't vanish if $|\lambda_{v}|<r^{\deg_{G}(v)}$ for all $v\in V(G)$.
\end{corollary}
\begin{proof}
Let $G=(V,E)$ with $|V|=n$.
Without loss of generality, assume $V=[n]$.
To use Lemma~\ref{lem:contraction}, we first need to create a ground set $[m]$. For each edge $e=\{u,v\}\in E$, let $u_{e}$ and $v_{e}$ be a copy of the vertex $u$ and $v$, respectively. Then consider the ground set $\bigcup_{e=\{u,v\}\in E}\{u_{e},v_{e}\}$, which has size $m := 2|E(G)|$. Let
$$P(\boldsymbol{z})=\prod_{e=\{u,v\}\in E}\left(\gamma z_{u_{e}}z_{v_{e}}+z_{u_{e}}+z_{v_{e}}+\beta\right).$$
Let $K=\{z\in\mathbb{C}:|z|\geq r\}$.
The assumption in the statement of the corollary guarantees that $P(\boldsymbol{z})$ vanishes only if some $z_\ell \in K$.

We can write $P(\boldsymbol{z})$ in the form from Lemma~\ref{lem:contraction} by defining a coefficient~$F(I)$ for
every subset~$I$ of the ground set. To do this, partition $E$
into sets $E_0$, $E_1$, and~$E_2$  
where $E_0$ is the set of $e=\{u,v\}$ such at $u_e$ and $v_e$ are both out of~$I$, $E_1$ is the set of $e=\{u,v\}$ with exactly one of $u_e,v_e$ in~$I$ and $E_2$ is the set of $e=\{u,v\}$ with both of $u_e$ and $v_e$ in~$I$. Then $F(I) = \gamma^{|E_2|} \beta^{|E_0|}$.

Now for each $v\in V$, let $I_v$ be the set of all ground set elements corresponding to vertex~$v$. That is, $I_v = \{v_{e} \mid e\in E, v\in e\}$.
Consider the polynomial 
$$Q(w_{1},\dots,w_{n}):=\sum_{J\subseteq V}F\left(\bigcup_{v\in J}I_{v}\right)\prod_{j\in J}w_{j}.$$
We can think of the set $J$ as the set of vertices with spin~$1$. Then $Q(\boldsymbol{w})=Z_{G}(\boldsymbol{w})$. So Lemma~\ref{lem:contraction}
guarantees that $Q(w_1,\ldots,w_n)$ vanishes only when, for some $v\in V$,
$w_v \in (-1)^{|I_v|+1} \prod_{i \in I_v} K$, proving the corollary.
 \end{proof}

\subsection[FPTAS]{Proof of \Cref{thm:FPTAS}}
In light of \Cref{cor:contractunitcircle} and Lemma~\ref{lem:zerofreeness}, it only remains to show zero-freeness for the single polynomial $\gamma z_{1}z_{2}+z_{1}+z_{2}+\beta$. 
\begin{lemma}\label{lem:unitcircle}
For real numbers $\beta,\gamma$ such that $\beta>\gamma$ and $\beta+\gamma> 2$, there exists $r>1$ such that the polynomial $\gamma z_{1}z_{2}+z_{1}+z_{2}+\beta$ doesn't vanish when $|z_{1}|,|z_{2}|<r$. 
\end{lemma}
\begin{proof}
In this proof, for any subset $S$ of the Riemann sphere $\mathbb{C}\cup\{\infty\}$, we use $S^{c}$ to denote its complement $\mathbb{C}\cup\{\infty\}\setminus S$.

In the degenerate case $\beta\gamma=1$, the polynomial $\gamma z_{1}z_{2}+z_{1}+z_{2}+\beta$ factorizes into $\gamma(z_{1}+\beta)(z_{2}+\beta)$, which clearly doesn't vanish when $|z_{1}|,|z_{2}|<\beta$. Since $\beta>1$, the conclusion holds in this case. In the following we assume $\beta\gamma\neq 1$, so $g:\mathbb{C}\cup\{\infty\}\rightarrow\mathbb{C}\cup\{\infty\}$ given by $z\mapsto -\frac{z+\beta}{\gamma z+1}$ is an invertible M\"{o}bius transformation on the Riemann sphere.

Let $D(0,r)$ denote the open disk $\{z\in\mathbb{C}:|z|<r\}$. Since $g(z)$ is the unique solution to the equation $\gamma z\cdot g(z)+z+g(z)+\beta=0$, it suffices to show for some $r>1$ that $g$ maps $D(0,r)$ into $D(0,r)^{c}$. 

Note that since $\beta,\gamma\in\mathbb{R}$, the transformation $g$ maps the $\mathbb{R}\cup\{\infty\}$ into $\mathbb{R}\cup\{\infty\}$. By conformality, $g$ maps any circle centered on $\mathbb{R}\cup\{\infty\}$ to a circle centered on $\mathbb{R}\cup\{\infty\}$. In particular, $g(D(0,r))$ is a disk centered on $\mathbb{R}\cup\{\infty\}$. So $g(D(0,r))$ and $D(0,r)$ are disjoint as long as their intersections with $\mathbb{R}\cup\{\infty\}$ are disjoint. It suffices to show that $g$ maps the real interval $(-r,r)$ into $(-r,r)^{c}$, for some $r>1$. By continuity of $g$ (as a map on the Riemann sphere), it also suffices to show that $g$ maps the interval $[-1,1]$ into $[-1,1]^{c}$.

Now take any real number $z\in [-1,1]$. From $\beta>\gamma$ and $\beta+\gamma>2$ we know $\beta>1$. we have
\begin{align*}
|g(z)|>1 &\Leftrightarrow |z+\beta|/|\gamma z+1|>1\\
&\Leftrightarrow (z+\beta)^{2}>(\gamma z+1)^{2} &(\text{since }\beta,\gamma,z\in\mathbb{R})\\
&\Leftrightarrow \left(1-z\frac{\gamma-1}{\beta-1}\right)\left(1+z\frac{\gamma+1}{\beta+1}\right)>0 &(\text{since }\beta>1).
\end{align*}

 It follows from $\beta>\gamma$ that $\frac{\gamma-1}{\beta-1}<1$ and $\frac{\gamma+1}{\beta+1}<1$, while it follows from $\beta+\gamma>2$ that $\frac{\gamma-1}{\beta-1}>-1$ and $\frac{\gamma+1}{\beta+1}>-1$. So both $|\frac{\gamma-1}{\beta-1}|$ and $|\frac{\gamma+1}{\beta+1}|$ are  less than 1. Since $|z|\leq 1$, we have

$$1-z\frac{\gamma-1}{\beta-1}>0\text{ and }1+z\frac{\gamma+1}{\beta+1}>0.$$
This proves $|g(z)|>1$ and hence $g$ maps the interval $[-1,1]$ into $[-1,1]^{c}$.
\end{proof}

\begin{corollary}\label{cor:negation}
For real numbers $\beta,\gamma$ such that $\beta<\gamma$ and $\beta+\gamma<-2$, there exists $r>1$ such that the polynomial $\gamma z_{1}z_{2}+z_{1}+z_{2}+\beta$ doesn't vanish when $|z_{1}|,|z_{2}|<r$.
\end{corollary}
\begin{proof}
By \Cref{lem:unitcircle}, the polynomial $(-\gamma)(-z_{1})(-z_{2})+(-z_{1})+(-z_{2})+(-\beta)$ doesn't vanish when $|z_{1}|,|z_{2}|<r$. So its negation, $\gamma z_{1}z_{2}+z_{1}+z_{2}+\beta$, doesn't vanish for $|z_{1}|,|z_{2}|<r$ either.
\end{proof}
Now we are ready to prove \Cref{thm:FPTAS}.

\begin{proof}[Proof of \Cref{thm:FPTAS}]
The range of parameters can be divided into 4 regions:

Case 1: $\beta>\gamma$ and $\beta+\gamma>2$. Combining \Cref{lem:unitcircle} and \Cref{cor:contractunitcircle}, there is a disk $D(0,r)$ containing 1 such that, for all graphs $G$ the polynomial $Z_{G}(x)$ doesn't vanish on $D(0,r)$. An FPTAS is thus given by \Cref{lem:zerofreeness}.

Case 2: $\beta<\gamma$ and $\beta+\gamma>2$. This case follows by symmetry from Case 1, as switching $\beta$ and $\gamma$ preserves $Z_{G}$.

Case 3: $\beta<\gamma$ and $\beta+\gamma<-2$. In a similar way to Case 1, this case follows by combining \Cref{cor:negation}, \Cref{cor:contractunitcircle} and \Cref{lem:zerofreeness}.

Case 4: $\beta>\gamma$ and $\beta+\gamma<-2$. This case follows by symmetry from Case 3.
\end{proof}

\subsection{Preliminaries for the FPRAS}\label{subsec:prelimFPRAS}

Our randomized approximation scheme for \Cref{thm:FPRAS} closely resembles the one in \cite{jerrum1993polynomial}. The first main ingredient in \cite{jerrum1993polynomial} is the ``subgraphs-world'' transformation that reduces a spin system problem to a Holant problem. Here, we need to use a slightly generalized version of the subgraphs-world transformation. Though it has appeared in various forms in the literature (e.g. \cite{guo2013complexity}), we introduce it here for the sake of completeness. 

\begin{definition}
Let $(\chi_{a,b})_{a,b\in \{0,1\}}$ be the standard Fourier characters on $\mathbb{F}_{2}^{2}$, defined by $\chi_{a,b}(x_{1},x_{2})=(-1)^{ax_{1}+bx_{2}}$, for $a,b,x_{1},x_{2}\in\{0,1\}$. For any function $\psi:\{0,1\}^{2}\rightarrow\mathbb{R}$, define its Fourier transform to be the function $\widehat{\psi}:\{0,1\}^{2}\rightarrow\mathbb{R}$ given by
$$\widehat{\psi}(a,b)=
\frac{1}{4} \sum_{x_{1},x_{2}\in \{0,1\}} \psi(x_{1},x_{2})\cdot \chi_{a,b}(x_{1},x_{2}),\quad\forall a,b\in\{0,1\}.$$
\end{definition}
Then we have the identity
$$\psi=\sum_{a,b\in\{0,1\}}\widehat{\psi}(a,b)\cdot \chi_{a,b}.$$
\begin{proposition}\label{prop:csptoholant}
Let $\psi:\{0,1\}^{2}\rightarrow\mathbb{Q}^{\geq 0}$. An FPRAS for $\mathsf{Holant}\left(\{\widehat{\psi}\}\cup\{\mathbf{Even}_{k}:k\geq 1\}\right)$ implies an FPRAS for $\mathsf{\#CSP}(\{\psi\})$.
\end{proposition}
\begin{proof}
Let $G=(V,E)$ be an instance of $\mathsf{\#CSP}(\{\psi\})$. Let $G'=(V',E')$ be defined by
$$V'=V\cup E\text{ and }E'=\bigcup_{e=\{i,j\}\in E}\{\{i,e\},\{j,e\}\}.$$ 
For every vertex $v\in V\subset V'$, let $F_{v}=\mathbf{Even}_{d}$, where $d:=\deg_{G}v$. For every vertex $e\in E\subset V'$, let $F_{e}=\widehat{\psi}$. In this way, we form a Holant instance $\phi$ with base graph $G'$. We have
\begin{align*}
Z_{G}&=\sum_{x\in\{0,1\}^{V}}\prod_{\{i,j\}\in E}\psi(x_{i},x_{j})\\
&=\sum_{x\in\{0,1\}^{V}}\prod_{\{i,j\}\in E}\sum_{a,b\in\{0,1\}}\widehat{\psi}(a,b)(-1)^{ax_{i}+bx_{j}}\\
&=\sum_{x\in\{0,1\}^{V}}\sum_{y\in\{0,1\}^{E'}}\prod_{e=\{i,j\}\in E}\widehat{\psi}(y_{i,e},y_{j,e})(-1)^{y_{i,e}x_{i}+y_{j,e}x_{j}}\\
&=\sum_{y\in \{0,1\}^{E'}}\left(\prod_{e=\{i,j\}\in E}\widehat{\psi}(y_{i,e},y_{j,e})\right)\left(\sum_{x\in \{0,1\}^{V}}\prod_{e=\{i,j\}\in E}(-1)^{y_{i,e}x_{i}+y_{j,e}x_{j}}\right)\\
&=\sum_{y\in\{0,1\}^{E'}}\left(\prod_{e=\{i,j\}\in E}\widehat{\psi}(y_{i,e},y_{j,e})\right)\left(\sum_{x\in\{0,1\}^{V}}\prod_{i\in V}(-1)^{x_{i}(\sum_{\{i,e\}\in E'}y_{i,e})}\right)\\
&=\sum_{y\in\{0,1\}^{E'}}\left(\prod_{e=\{i,j\}\in E}\widehat{\psi}(y_{i,e},y_{j,e})\right)\left(\prod_{i\in V}\left(1+(-1)^{(\sum_{\{i,e\}\in E'}y_{i,e})}\right)\right)\\
&=2^{|V|}\sum_{y\in\{0,1\}^{E'}}\prod_{e=\{i,j\}\in E}\widehat{\psi}(y_{i,e},y_{j,e})\prod_{i\in V}\mathbf{Even}\left((y_{i,e})_{\{i,e\}\in E'}\right).\\
&=2^{|V|}[\![ \phi ]\!],
\end{align*}
where $[\![ \phi ]\!]$ denotes the partition function of the Holant instance $\phi$.
\end{proof}

In \cite{jerrum1993polynomial}, the next step is to prove the rapid mixing of a Markov chain associated to the Holant problem and compute the partition function using an MCMC algorithm. But fortunately for us, we don't even need to define the Markov chain, as the powerful framework of \cite{mcquillan2013approximating} has reduced all these efforts to verifying some simple criteria:

\begin{definition}
  \label{def:wind}
  For any finite set $J$ and any configuration $x\in \{0,1\}^J$,
  define $\mathcal{M}_x$ to be the set of partitions of $\left\{i\middle| x_i=1\right\}$
  into pairs and at most one singleton.  A function
  $F: \{0,1\}^J\rightarrow \mathbb{Q}^{\geq 0} $ is \textbf{windable} if
  there exist values $B(x,y,M)\geq 0$ for all $x,y\in \{0,1\}^J$ and
  all $M\in \mathcal{M}_{x\oplus y}$ (where $x\oplus y$ stands for the bit-wise XOR of $x$ and $y$) satisfying:
  \begin{enumerate}
  \item $F(x)F(y)=\sum_{M\in \mathcal{M}_{x\oplus y}}B(x,y,M)$ for all
    $x,y\in \{0,1\}^J$, and
  \item $B(x,y,M)=B(x\oplus S,y\oplus S,M)$ for all $x,y\in \{0,1\}^J$
    and all $S\in M\in \mathcal{M}_{x\oplus y}$.
  \end{enumerate}
  Here $x\oplus S$ denotes the vector obtained by changing $x_i$ to
  $1-x_i$ for the one or two elements $i$ in $S$.
\end{definition}
\begin{lemma}\label{lem:smallwindable}
Any function $\{0,1\}^{2}\rightarrow\mathbb{Q}^{\geq 0}$ is windable.
\end{lemma}
\begin{proof}
The statement follows directly by combining Lemma 7 and Lemma 15 in \cite{mcquillan2013approximating}.
\end{proof}
\begin{definition}
A function $F:\{0,1\}^{J}\mapsto\mathbb{Q}^{\geq 0}$ is \textbf{strictly terraced} if
\[ F(x)=0 \implies F(x\oplus e_{i})=F(x\oplus e_{j})\qquad\text{ for all $x\in\{0,1\}^J$ and all $i,j\in J$}.\]
Here $x\oplus e_{i}$ denotes the vector obtained by
changing $x_i$ to $1-x_i$.
\end{definition}
\begin{lemma}[Theorem 4 in \cite{mcquillan2013approximating}]\label{thm:windableterraced}
If $\mathcal{F}$ is a finite class of strictly terraced windable functions, then there is an FPRAS for $\mathsf{Holant}(\mathcal{F})$.
\end{lemma}
\begin{corollary}\label{cor:windableterraced+}
If $\mathcal{F}$ is a finite class of strictly terraced windable functions, then there is an FPRAS for $\mathsf{Holant}(\mathcal{F}\cup\{\mathbf{Even}_{k}:k\geq 1\})$.
\end{corollary}
\begin{proof}
For $k>3$, an $\mathbf{Even}_{k}$ constraint can easily be realized using a path of $(k-2)$   $\mathbf{Even}_{3}$ constraints (using $k-3$ additional variables).
For example, 
\[\mathbf{Even}_5(x_1,x_2,x_3,x_4,x_5) = \sum_{y_{1},y_{2}\in \{0,1\}}\mathbf{Even}_3(x_1,x_2,y_1) \mathbf{Even}_3(y_1,x_3,y_2) \mathbf{Even}_3(y_2,x_4,x_5).\] Thus,
it suffices to show that there is an FPRAS for $\mathsf{Holant}(\mathcal{F}\cup\{\mathbf{Even}_{3}\})$. Since $\mathbf{Even}_{3}$ is windable (see \cite[Lemma 17]{mcquillan2013approximating}) and strictly terraced, the claim follows from \Cref{thm:windableterraced}.
\end{proof}

\subsection[FPRAS]{Proof of \Cref{thm:FPRAS}} \label{subsec:proofFPRAS}
Now, it suffices to verify that certain constraint functions are windable and strictly terraced. 
\begin{lemma}\label{lem:lowerright}
For rational numbers $\beta,\gamma$ such that $\beta>\gamma$ and $\beta+\gamma\geq 2$, the function $\psi:\{0,1\}^{2}\rightarrow\mathbb{Q}$ defined by $\begin{bmatrix}
\psi(0,0) & \psi(0,1)\\
\psi(1,0) & \psi(1,1)
\end{bmatrix}=\begin{bmatrix}
\beta & 1\\
1 & \gamma
\end{bmatrix}$ satisfies the property that $\widehat{\psi}$ is windable and strictly terraced. 
\end{lemma}
\begin{proof}
Since $\begin{bmatrix}
\widehat{\psi}(0,0) & \widehat{\psi}(0,1)\\
\widehat{\psi}(1,0) & \widehat{\psi}(1,1)
\end{bmatrix}=\dfrac{1}{4}\begin{bmatrix}
\beta+\gamma+2 & \beta-\gamma\\
\beta-\gamma & \beta+\gamma-2
\end{bmatrix}$, we have $\widehat{\psi}(x)\geq 0$ for all $x\in\{0,1\}^{2}$, and the only possibility of $\widehat{\psi}(x)=0$ is when $\beta+\gamma=2$ and $x=(1,1)$. In that case, we have $\widehat{\psi}(1,0)=\widehat{\psi}(0,1)=\frac{\beta-\gamma}{4}$. It follows that $\widehat{\psi}$ is strictly terraced. 

The windablity of $\widehat{\psi}$ follows from \Cref{lem:smallwindable}.
\end{proof}

\begin{lemma}\label{lem:upperleft}
For rational numbers $\beta,\gamma$ such that $\beta<\gamma$ and $\beta+\gamma\leq -2$, the function $\psi:\{0,1\}^{2}\rightarrow\mathbb{Q}$ defined by $\begin{bmatrix}
\psi(0,0) & \psi(0,1)\\
\psi(1,0) & \psi(1,1)
\end{bmatrix}=\begin{bmatrix}
\beta & 1\\
1 & \gamma
\end{bmatrix}$ satisfies the property that $-\widehat{\psi}$ is windable and strictly terraced. 
\end{lemma}
\begin{proof}
Since $\begin{bmatrix}
-\widehat{\psi}(0,0) & -\widehat{\psi}(0,1)\\
-\widehat{\psi}(1,0) & -\widehat{\psi}(1,1)
\end{bmatrix}=\dfrac{1}{4}\begin{bmatrix}
-2-\beta-\gamma & \gamma-\beta\\
\gamma-\beta & 2-\beta-\gamma
\end{bmatrix}$, we have $-\widehat{\psi}(x)\geq 0$ for all $x\in\{0,1\}^{2}$, and the only possibility of $-\widehat{\psi}(x)=0$ is when $\beta+\gamma=-2$ and $x=(0,0)$. In that case, we have $-\widehat{\psi}(1,0)=\widehat{\psi}(0,1)=\frac{\gamma-\beta}{4}$. It follows that $-\widehat{\psi}$ is strictly terraced. 

The windablity of $-\widehat{\psi}$ follows from \Cref{lem:smallwindable}.
\end{proof}
Now we are ready to prove \Cref{thm:FPRAS}.
\begin{proof}[Proof of \Cref{thm:FPRAS}]
The range of parameters can be divided into 4 regions:

Case 1: $\beta>\gamma$ and $\beta+\gamma\geq 2$. Combining \Cref{lem:lowerright} and \Cref{cor:windableterraced+}, there is an FPRAS for $\mathsf{Holant}\left({\widehat{\psi}}\cup\{\mathbf{Even}_{k}:k\geq 1\}\right)$, where $\psi:\{0,1\}^{2}\rightarrow\mathbb{Q}$ defined by $\begin{bmatrix}
\psi(0,0) & \psi(0,1)\\
\psi(1,0) & \psi(1,1)
\end{bmatrix}=\begin{bmatrix}
\beta & 1\\
1 & \gamma
\end{bmatrix}$. An FPRAS for $\mathsf{\#CSP}(\{\psi\})$ is thus given by \Cref{prop:csptoholant}.

Case 2: $\beta<\gamma$ and $\beta+\gamma\geq 2$. This case follows by symmetry from Case 1, as switching $\beta$ and $\gamma$ preserves $Z_{G}$.

Case 3: $\beta<\gamma$ and $\beta+\gamma \leq -2$. In a similar way to Case 1, this case follows by combining \Cref{lem:upperleft}, \Cref{cor:windableterraced+} and \Cref{prop:csptoholant}.

Case 4: $\beta>\gamma$ and $\beta+\gamma \leq -2$. This case follows by symmetry from Case 3.
\end{proof}

\section{The Recursion Method}\label{sec:five}

This section collects the proofs of \Cref{thm:positivity,thm:diskaround0,thm:notthreshold}. All of three proofs are based on the recursion method, but they recurse with different subsets of the complex plane.  

\label{sec:recursion}
\subsection{Recursion with Real Intervals}
\label{subsec:positivity}

\Cref{prop:dense} shows that when $\gamma<0$ and the parameter point $(\beta,\gamma)$ lies slightly to the left of the line $\beta+\gamma=1$, the realizable ratios are dense in $\mathbb{R}$. The next lemma says that, if $(\beta,\gamma)$ lie to the right of the line $\beta+\gamma=1$, the realizable ratios are bounded by an interval, even if certain external fields are allowed in the system. This is in stark contrast to \Cref{prop:dense}, and indicates that some phase transition happens at the line $\beta+\gamma=1$. As a byproduct, the next lemma also shows that the partition function is always positive, in contrast with \Cref{thm:sharpPhard}. 
\begin{lemma}\label{lem:positivity}
Let $\beta,\gamma$ be real numbers such that $\beta+\gamma\geq 1$ and $\gamma<0$. Then for any graph $G$, any external field $\bflam\in[\frac{\gamma}{\beta},1]^{V(G)}$ and any $v\in V(G)$, 
\begin{enumerate}[label = (\arabic*)]
\setlength\itemsep{0pt}
\item the ratio $R_{G,v}(\bflam)$ is well defined and falls in the interval $[\frac{\gamma}{\beta},1]$;
\item the partition function $Z_{G}(\bflam)$ is positive.
\end{enumerate}
\end{lemma}
\begin{proof}
Since $\gamma<0$ and $\beta\geq 1-\gamma>-\gamma$, we have $\frac{\gamma}{\beta}\in(-1,0)$. Observe that a self-loop on a vertex $v$ has the same effect as multiplying the local external field $\lambda_{v}$ by $\frac{\gamma}{\beta}$
and multiplying the partition function by $\beta>0$.
The new local external field would still be in $[\frac{\gamma}{\beta},1]^{V(G)}$ because $\frac{\gamma}{\beta}\in(-1,0)$. So we may assume $E(G)$ doesn't contain self-loops.

We then perform induction on $|V(G)|+|E(G)|$ for the two statements together. The base case is where $V(G)$ is a singleton $\{v\}$ and $E(G)=\emptyset$, in which case $R_{G,v}(\bflam)=\lambda_{v}$ and $Z_{G}(\bflam)=1+\lambda_{v}>0$. Now assume $G$ is a graph such that $|V(G)|+|E(G)|\geq 2$. Assume also that the induction hypotheses (1) and (2) hold for all graphs with a smaller combined number of vertices and edges. Consider a vertex $v\in V(G)$. The induction step is to prove statements (1) and (2) for the pair $(G,v)$.

If no edge is incident to $v$, let $G-v$ be the graph obtained from $G$ by deleting the vertex $v$, and let $\bflamp$ be $\bflam$ restricted on $V(G)\setminus\{v\}$. We have $[Z_{G,v}(\bflam)]_{0}=Z_{G-v}(\bflamp)>0$ from the induction hypothesis (2) applied on $G-v$, so $R_{G,v}$ is well-defined. Clearly $R_{G,v}=\lambda_{v}\in[\frac{\gamma}{\beta},1]$ and $Z_{G}(\bflam)=(1+\lambda_{v})Z_{G-v}(\bflamp)>0$, completing the induction step. In the following, we deal with the harder case where there is an edge $\{u,v\}$ incident to $v$.

Define $G_{1}=G-\{u,v\}$, the graph obtained from $G$ by deleting the edge $\{u,v\}$. The following equations are immediate consequences:
\begin{align*}
[Z_{G,v}(\bflam)]_{0}&=\beta\cdot \left[Z_{G_{1},u,v}(\bflam)\right]_{0,0}+\left[Z_{G,u,v}(\bflam)\right]_{1,0},\\
[Z_{G,v}(\bflam)]_{1}&= 
\left[Z_{G_{1},u,v}(\bflam)\right]_{0,1}+\gamma\cdot\left[Z_{G,u,v}(\bflam)\right]_{1,1}.
\end{align*}
To shorten expressions, let $A=\left[Z_{G_{1},u,v}(\bflam)\right]_{0,0}$, $B=\left[Z_{G_{1},u,v}(\bflam)\right]_{1,0}$, $C=\left[Z_{G_{1},u,v}(\bflam)\right]_{0,1}$, and $D=\left[Z_{G_{1},u,v}(\bflam)\right]_{1,1}$. So the ratio $R_{G,v}(\bflam)=\dfrac{[Z_{G,v}(\bflam)]_{1}}{[Z_{G,v}(\bflam)]_{0}}$ can be written as $\frac{C+\gamma D}{\beta A+B}$. The entire goal of the remaining proof is to use the induction hypothesis to show that
\begin{equation}\label{eq:goal}
\beta A+B> 0\quad\text{and}\quad\frac{\gamma}{\beta}\leq \frac{C+\gamma D}{\beta A+B}\leq 1.
\end{equation}
Clearly proving \eqref{eq:goal} would imply that the first statement holds for the pair $(G,v)$. The second statement would also follow, since we would have $Z_{G}(\bflam)=\beta A+B+C+\gamma D\geq(1+\frac{\gamma}{\beta})(\beta A+B)>0$. 

To prove \eqref{eq:goal}, we make use of the induction hypotheses in 6 different ways:
\begin{itemize}
\setlength\itemsep{0pt}
\item First, since $|V(G_{1})|+|E(G_{1})|=|V(G)|+|E(G)|-1$, the induction hypothesis can be applied to $G_{1}$. Since 
$R_{G_{1},v}(\bflam)=\dfrac{\left[Z_{G_{1},v}(\bflam)\right]_{1}}{\left[Z_{G_{1},v}(\bflam)\right]_{0}}=\dfrac{C+D}{A+B}$,
the induction hypothesis (1) gives
\begin{equation}\label{eq:G1}
\frac{\gamma}{\beta}\leq\frac{C+D}{A+B}\leq 1.
\end{equation}
\item If we define $\bflamp\in[\frac{\gamma}{\beta},1]^{V(G)}$ to be $\lambda'_{u}=\frac{\gamma}{\beta}\lambda_{u}$ and $\lambda'_{x}=\lambda_{x}$ for all $x\neq u$, we would have $R_{G_{1},v}(\bflamp)=\dfrac{\left[Z_{G_{1},u,v}(\bflamp)\right]_{0,1}+\left[Z_{G_{1},u,v}(\bflamp)\right]_{1,1}}{\left[Z_{G_{1},u,v}(\bflamp)\right]_{0,0}+\left[Z_{G_{1},u,v}(\bflamp)\right]_{1,0}}=\dfrac{C+\frac{\gamma}{\beta}D}{A+\frac{\gamma}{\beta}B}=\dfrac{\beta C+\gamma D}{\beta A+\gamma B}$. The induction hypothesis (1) gives
\begin{equation}\label{eq:G2}
\frac{\gamma}{\beta}\leq\frac{\beta C+\gamma D}{\beta A+\gamma B}\leq 1.
\end{equation}

\item Similarly we can define $\bflamp\in[\frac{\gamma}{\beta},1]^{V(G)}$ to be $\lambda'_{v}=\frac{\gamma}{\beta}\lambda_{v}$ and $\lambda'_{x}=\lambda_{x}$ for all $x\neq v$. The induction hypothesis on $R_{G_{1},u}(\bflamp)$ gives
\begin{equation}\label{eq:G3}
\frac{\gamma}{\beta}\leq\frac{\beta B+\gamma D}{\beta A+\gamma C}\leq 1.
\end{equation}
\item Consider $G_{2}=G/\{u,v\}$. This means contracting the edge $\{u,v\}$: create a new vertex $w$, delete $\{u,v\}$ from $E(G)$, and in every other member of $E(G)$, substitute $w$ for any appearance of $u$ and $v$ as endpoints. If we define $\bflamp\in[\frac{\gamma}{\beta},1]^{V(G_{2})}$ to be $\lambda'_{w}=\lambda_{u}\lambda_{v}$ and $\lambda'_{x}=\lambda_{x}$ for all $x\neq w$, it is clear that $\left[Z_{G_{2},w}(\bflamp)\right]_{0}=A$ and $\left[Z_{G_{2},w}(\bflamp)\right]_{1}=D$. Using the induction hypothesis on $G_{2}$, we get
\begin{equation}\label{eq:G4}
\frac{\gamma}{\beta}\leq \frac{D}{A}\leq 1.
\end{equation}
\item If we define $\bflamp\in[\frac{\gamma}{\beta},1]^{V(G)}$ to be $\lambda'_{u}=0$ and $\lambda'_{x}=\lambda_{x}$ for all $x\neq u$, this has the effect of ``pinning'' the spin on vertex $u$ to 0. We thus have $\left[Z_{G_{1},v}(\bflamp)\right]_{0}=A$ and $\left[Z_{G_{1},v}(\bflamp)\right]_{1}=C$. Using the induction hypothesis (1) on $G_{1}$, we have $\frac{\gamma}{\beta}\leq \frac{C}{A}\leq 1$. Similarly, by pinning the spin of $v$ to 0, we get $\frac{\gamma}{\beta}\leq \frac{B}{A}\leq 1$.

\item 
If we define $\bflamp$ to be $\lambda'_{u}=\lambda'_{v}=0$ and $\lambda'_{x}=\lambda_{x}$ for all $x\neq u,v$, we can pin the spins of both $u$ and $v$ to 0. Then $A=Z_{G_{1}}(\bflamp)>0$, by the induction hypothesis (2) on $G_{1}$. Therefore, using $\frac{\gamma}{\beta}\leq \frac{B}{A}\leq 1$, we get $A+B\geq (1+\frac{\gamma}{\beta})A>0$ and $\beta A+\gamma B\geq(\beta+\gamma)A>0$. Similarly we also have $\beta A+\gamma C>0$. In conclusion,
\begin{equation}\label{eq:samesign}
A+B,\quad \beta A+\gamma B,\quad \beta A+\gamma C \text{ and } A \text{ are all positive.}
\end{equation}
\end{itemize}

We claim that given the condition \eqref{eq:samesign}, the inequalities \eqref{eq:G1}, \eqref{eq:G2}, \eqref{eq:G3} and \eqref{eq:G4} together imply \eqref{eq:goal}. This is, in fact, purely elementary algebra. Recall that $\beta\geq 1-\gamma>1$. So the first half of \eqref{eq:goal} is obvious: $\beta A+B=(\beta-1)A+(A+B)$ is also positive. Now, using the four inequalities \eqref{eq:G1} through \eqref{eq:G4}, we have
\begin{equation}\label{eq:right}
1-\frac{C+\gamma D}{\beta A+B}=\frac{A+B}{\beta A+B}\left(1-\frac{C+D}{A+B}\right)+\frac{\beta A}{\beta A+B}\left(\frac{D}{A}-\frac{\gamma}{\beta}\right)+\frac{(\beta+\gamma-1)A}{\beta A+B}\left(1-\frac{D}{A}\right)\geq 0,
\end{equation}
and
\begin{equation}\label{eq:left}
\begin{split}
\frac{C+\gamma D}{\beta A+B}-\frac{\gamma}{\beta}=\,&\frac{(\beta^{2}+\gamma\beta+\gamma^{2})}{(\beta+\gamma)(\beta^{2}+\gamma^{2})}\cdot\frac{\beta A+\gamma B}{\beta A+B}\left(\frac{\beta C+\gamma D}{\beta A+\gamma B}-\frac{\gamma}{\beta}\right)+\\
&\frac{-\gamma\beta}{(\beta+\gamma)(\beta^{2}+\gamma^{2})}\cdot\frac{\beta A+\gamma C}{\beta A+B}\left(\frac{\beta B+\gamma D}{\beta A+\gamma C}-\frac{\gamma}{\beta}\right)+\\
&\frac{-\gamma(\beta+\gamma-1)}{(\beta+\gamma)}\cdot\frac{A}{\beta A+B}\left(1-\frac{D}{A}\right)\geq 0.
\end{split}
\end{equation}
This completes the proof of the second half of \eqref{eq:goal}. Note that both \eqref{eq:right} and \eqref{eq:left} crucially rely on the condition $\beta+\gamma\geq 1$. 
\end{proof}
\Cref{thm:positivity} then follows almost immediately: 
\begin{proof}[Proof of \Cref{thm:positivity}]
If both $\beta$ and $\gamma$ is nonnegative, it is clear that $Z_{G}>0$. So by symmetry we can assume $\gamma$ is negative, and then the theorem follows from the statement (2) of \Cref{lem:positivity}, since $Z_{G}=Z_{G}(\mathbf{1})$.
\end{proof}

\subsection{Recursion with Circular Regions}
\label{subsec:circular}

The one major range of parameters where approximation complexity is unsettled is where $1\leq \beta+\gamma< 2$ and (without loss of generality) $\gamma<0$. As with the proof of \Cref{thm:FPTAS}, zero-freeness is still one of the most natural approaches to try, in terms of proving approximation efficiency. 

However, in this range, the contraction method doesn't seem to apply easily. Instead, we will try to attack this range using another important method: induction on the number of vertices, whose power is already well demonstrated in \Cref{subsec:positivity}, where we proved \Cref{thm:positivity}. In fact, \Cref{thm:positivity} itself can be viewed as a ``zero-freeness'' result: it implies that the partition function $Z_{G}(x)$ is zero-free on the real interval $[\frac{\gamma}{\beta},1]$. The only weakness is, in order to make use of \Cref{lem:zerofreeness} we must prove a zero-free neighborhood of $[0,1]$ on \textit{the complex plane}. 

Unfortunately, as mentioned in \Cref{subsec:results}, we are unable to do this for the whole range $\{(\beta,\gamma):\gamma<0\text{ and }1\leq \beta+\gamma\leq 2\}$ . In general, it appears challenging to prove optimal zero-free regions in the complex plane (c.f. \cite{guo2020zeros,bencs2023complex}). In this section, we will prove \Cref{thm:diskaround0}, which gives an optimal \textit{circular} zero-free region.

\begin{lemma}\label{lem:diskaround0}
Let $\beta,\gamma$ be real numbers such that $\gamma<0$ and $1\leq \beta+\gamma\leq 2$. Let $r=\frac{\beta-1}{1-\gamma}\in (0,1]$ and let $K$ denote the region $\{z\in\mathbb{C}:|z|\leq r\}\setminus\{-1\}$ (remark: when $\beta+\gamma<2$, the exclusion of $\{-1\}$ is redundant since $r<1$). Then for any external field $\bflam\in K^{V(G)}$ and any $v\in V(G)$,
\begin{enumerate}[label = (\arabic*)]
\setlength\itemsep{0pt}
\item the ratio $R_{G,v}(\bflam)$ is well defined and falls in $K$;
\item the partition function $Z_{G}(\bflam)$ is nonzero.
\end{enumerate}
\end{lemma}
\begin{remark}
The proof below follows the general structure of the proof of \Cref{thm:positivity}. The main differences are in the ways we use the induction hypotheses. Although the first half of the proofs are mostly identical, there are occasionally minor differences. So we still present the complete proof.
\end{remark}
\begin{proof}[Proof of \Cref{lem:diskaround0}]
 Recall that a self-loop on a vertex $v$ has the same effect as multiplying the local external field $\lambda_{v}$ by $\frac{\gamma}{\beta}$, which doesn't change the fact that $\bflam\in K^{V(G)}$ because $\frac{\gamma}{\beta}\in(-1,0)$. So we may assume $E(G)$ doesn't contain self-loops.

We then perform induction on $|V(G)|+|E(G)|$ for the two statements together. The base case is where $V(G)$ is a singleton $\{v\}$ and $E(G)=\emptyset$, in which case $R_{G,v}(\bflam)=\lambda_{v}$ and $Z_{G}(\bflam)=1+\lambda_{v}\neq 0$ (this is where we need the exclusion of $-1$). Now assume $G$ is a graph such that $|V(G)|+|E(G)|\geq 2$. Assume also that the induction hypotheses (1) and (2) hold for all graphs with a smaller combined number of vertices and edges. Consider a vertex $v\in V(G)$. The induction step is to prove statements (1) and (2) for the pair $(G,v)$.

If no edge is incident to $v$, let $G-v$ be the graph obtained from $G$ by deleting the vertex $v$, and let $\bflamp$ be $\bflam$ restricted on $V(G)\setminus\{v\}$. We have $[Z_{G,v}(\bflam)]_{0}=Z_{G-v}(\bflamp)\neq 0$ from the induction hypothesis (2) applied on $G-v$, so $R_{G,v}$ is well-defined. Clearly $R_{G,v}=\lambda_{v}\in K$ and $Z_{G}(\bflam)=(1+\lambda_{v})Z_{G-v}(\bflamp)\neq 0$, completing the induction step. In the following, we deal with the harder case where there is an edge $\{u,v\}$ incident to $v$.

Define $G_{1}=G-\{u,v\}$, the graph obtained from $G$ by deleting the edge $\{u,v\}$. We have:
\begin{align*}
[Z_{G,v}(\bflam)]_{0}&=\beta\cdot \left[Z_{G_{1},u,v}(\bflam)\right]_{0,0}+\left[Z_{G,u,v}(\bflam)\right]_{1,0},\\
[Z_{G,v}(\bflam)]_{1}&= 
\left[Z_{G_{1},u,v}(\bflam)\right]_{0,1}+\gamma\cdot\left[Z_{G,u,v}(\bflam)\right]_{1,1}.
\end{align*}
To shorten expressions, let $A=\left[Z_{G_{1},u,v}(\bflam)\right]_{0,0}$, $B=\left[Z_{G_{1},u,v}(\bflam)\right]_{1,0}$, $C=\left[Z_{G_{1},u,v}(\bflam)\right]_{0,1}$, and $D=\left[Z_{G_{1},u,v}(\bflam)\right]_{1,1}$. So the ratio $R_{G,v}(\bflam)=\dfrac{[Z_{G,v}(\bflam)]_{1}}{[Z_{G,v}(\bflam)]_{0}}$ can be written as $\frac{C+\gamma D}{\beta A+B}$. 

We make use of the induction hypotheses in the following 3 ways:
\begin{itemize}
\setlength\itemsep{0pt}
\item If we define $\bflamp\in K^{V(G)}$ to be $\lambda'_{v}=0$ and $\lambda'_{x}=\lambda_{x}$ for all $x\neq v$, we would have
$R_{G_{1},u}(\bflamp)=\dfrac{\left[Z_{G_{1},u}(\bflamp)\right]_{1}}{\left[Z_{G_{1},u}(\bflamp)\right]_{0}}=\dfrac{B}{A}$. The induction hypothesis (1) gives $\frac{B}{A}\in K$.
\item If we define $\bflamp\in K^{V(G)}$ to be $\lambda'_{u}=0$ and $\lambda'_{x}=\lambda_{x}$ for all $x\neq u$, we would have
$R_{G_{1},v}(\bflamp)=\dfrac{\left[Z_{G_{1},v}(\bflamp)\right]_{1}}{\left[Z_{G_{1},v}(\bflamp)\right]_{0}}=\dfrac{C}{A}$. The induction hypothesis (1) gives $\frac{C}{A}\in K$.
\item Consider $G_{2}=G/\{u,v\}$. This means contracting the edge $\{u,v\}$: create a new vertex $w$, delete $\{u,v\}$ from $E(G)$, and in every other member of $E(G)$, substitute $w$ for any appearance of $u$ and $v$ as endpoints. If we define $\bflamp\in K^{V(G_{2})}$ to be $\lambda'_{w}=\lambda_{u}$ and $\lambda'_{x}=\lambda_{x}$ for all $x\neq w$, using the induction hypothesis on $G_{2}$, we get $\dfrac{\left[Z_{G_{2},w}(\bflamp)\right]_{1}}{\left[Z_{G_{2},w}(\bflamp)\right]_{0}}\in K$. Now, we still have $\left[Z_{G_{2},w}(\bflamp)\right]_{0}=A$, but $D=\lambda_{v}\cdot \left[Z_{G_{2},w}(\bflamp)\right]_{1}$. So $\frac{D}{A}=\lambda_{v}\dfrac{\left[Z_{G_{2},w}(\bflamp)\right]_{1}}{\left[Z_{G_{2},w}(\bflamp)\right]_{0}}\in K\cdot K$, where $K\cdot K$ stands for the Minkowski product of $K$ with itself.
\end{itemize}

Now we are ready to complete the induction step. Recall that $\beta\geq 1-\gamma>1$. It follows from $\frac{B}{A}\in K$ and $-\beta\not\in K$ that $\beta A+B\neq 0$. So $R_{G,v}(\bflam)=\frac{C+\gamma D}{\beta A+B}$ is well defined. What's more,
\begin{align}
\label{eq:ABopposite}
r\cdot |\beta A+B| &\geq r\beta |A|-r|B|\\
\label{eq:Blarge}
&\geq r\beta |A|-r^{2}|A| &(\text{since }\tfrac{B}{A}\in K) \\
& = r(1-\gamma r)|A| &(\text{since }r=\tfrac{\beta-1}{1-\gamma})\nonumber\\
& \geq |C|+(-\gamma)|D| &(\text{since }\tfrac{C}{A}\in K\text{ and }\tfrac{D}{A}\in K\cdot K)\\
& \geq |C+\gamma D|,
\end{align}
so $R_{G,v}(\bflam)=\frac{C+\gamma D}{\beta A+B}\in \{z\in \mathbb{C}:|z|\leq r\}$. The only possibility of $R_{G,v}(\bflam)=-1$ is when $r=1$ and all the inequalities above hold with equality. But then the equalities in \eqref{eq:ABopposite} and \eqref{eq:Blarge} together imply that $B=-rA=-A$, violating the induction hypothesis $\frac{B}{A}\in K$. So we conclude that $R_{G,v}(\bflam)\in \{z\in \mathbb{C}:|z|\leq r\}\setminus\{-1\}=K$, proving statement (1) for the pair $(G,v)$. Finally, from $R_{G,v}(\bflam)=\frac{C+\gamma D}{\beta A+B}\neq -1$ it immediately follows that $Z_{G}(\bflam)=\beta A+B+C+\gamma D=\left(1+\frac{C+\gamma D}{\beta A+B}\right)(\beta A+B)\neq 0$, proving statement (2) for the pair $(G,v)$.
\end{proof}

We also give a complementary result showing that the radius $r=\frac{\beta-1}{1-\gamma}$ in \Cref{lem:diskaround0} is optimal:

\begin{lemma}\label{lem:optimalradius}
Let $\beta,\gamma$ be real numbers such that $\gamma <0$ and $1\leq \beta+\gamma\leq 2$. For any $r>\frac{\beta-1}{1-\gamma}$, there exists a graph $G$ such that the polynomial $Z_{G}(x)$ has a root in the disk $\{z\in\mathbb{C}:|z|<r\}$.
\end{lemma}
\begin{proof}
Consider the graph $G_{n}$ with $V(G_{n})=\{v_{0},v_{1},\dots,v_{n}\}$ and 
$$E(G_{n})=\{\{v_{0},v_{1}\},\{v_{0},v_{2}\},\dots,\{v_{0},v_{n}\}\}$$ (a star graph). It's easy to compute that $Z_{G_{n}}(x)=(\beta +x)^{n}+x(1+\gamma x)^{n}$. For any $\frac{\beta-1}{1-\gamma}<r<\beta$, we have $\frac{1-\gamma r}{\beta -r}>1$. So for sufficiently large $n$,  $Z_{G_{n}}(-r)=(\beta-r)^{n}(1-r(\frac{1-\gamma r}{\beta -r})^{n})<0$. But $Z_{G_{n}}(0)=\beta^{n}>0$, so it follows from the intermediate value theorem that the polynomial $Z_{G_{n}}(x)$ has a root in the real interval $(-r,0)$.
\end{proof}
\begin{proof}[Proof of \Cref{thm:diskaround0}]
The zero-freeness follows from statement (2) of \Cref{lem:diskaround0}, and the optimality of the radius follows from \Cref{lem:optimalradius}.
\end{proof}

\subsection{Recursion with Uncentered Circular Regions}

In this section, we present a proof of \Cref{thm:notthreshold}. The idea is similar to the proof of \Cref{thm:diskaround0} in \Cref{subsec:circular}, but this time we recurse with a circular region that's \textit{not} centered at 0. The main new ingredient is that we treat isolated vertices and non-isolated vertices of the graph separately through casework.
\label{subsec:uncentered}
\begin{lemma}\label{lem:uncenteredrecur}
Let $g:(1,+\infty)\rightarrow (0,1)$ be the following function:
\begin{equation}\label{eq:defofg}
g(\beta)=\max\left\{\frac{\beta-2}{\beta^{2}-1},\frac{(\beta-1)^{2}}{\beta^{3}+\beta^{2}-\beta}\right\}.
\end{equation}
For fixed real parameters $\beta,\gamma$ with $\gamma<0$ and $\beta+\gamma>2-g(\beta)$, there exists an open neighborhood $U$ of $[\frac{\gamma}{\beta},1]$ on the complex plane and a closed disk $K\subset\mathbb{C}$ such that for any graph $G$ and any external field $\bflam\in U^{V(G)}$,  
\begin{enumerate}[label = (\arabic*)]
\setlength\itemsep{0pt}
\item  the ratio $R_{G,v}(\bflam)$ is well defined and falls in $K$ for any non-isolated vertex $v\in V(G)$;
\item the partition function $Z_{G}(\bflam)$ is nonzero.
\end{enumerate}
\end{lemma}
Before proving the lemma, we first need to specify the regions $U$ and $K$. Let $U=\{z\in \mathbb{C}:\exists x\in [\frac{\gamma}{\beta},1]\text{ s.t. }|z-x|< \varepsilon\}$, where $\varepsilon>0$ is a sufficiently small constant. Let $K$ be the closed disk with the real interval $[a,b]$ as its diameter, where $a\in(-1,0)$ and $b\in(0,+\infty)$ are constants to be determined later. In fact, during the proof of the lemma, we will impose several requirements on $a$ and $b$, and in the end we will show that these requirements are jointly satisfiable in the parameter range $\{(\beta,\gamma):\gamma<0\text{ and }\beta+\gamma>2-g(\beta)\}$. 

\begin{proof}[Proof of \Cref{lem:uncenteredrecur}]
Since $\gamma<0$ and $\beta\geq 2-\gamma-g(\beta)>-\gamma$, we have $\frac{\gamma}{\beta}\in (-1,0)$. Given the choice of the region $U$, it is clear that multiplying any component of $\bflam$ by $\frac{\gamma}{\beta}$ doesn't change the fact that $\bflam\in U^{V(G)}$. So we may assume $E(G)$ doesn't contain self-loops.

We then perform induction on $|E(G)|$ for the two statements together. The base case is where $E(G)=\emptyset$. In this case all vertices are isolated, and statement (1) holds vacuously. As for statement (2), we have $Z_{G}(\bflam)=\prod_{v\in V(G)}(1+\lambda_{v})\neq 0$, since $-1\not\in U$ provided that $\varepsilon$ is sufficiently small. Now assume $G$ is a graph with $|E(G)|\geq 1$. Assume also that the induction hypotheses (1) and (2) hold for all graphs with a smaller number of edges. Consider a non-isolated vertex $v\in V(G)$. The induction step is to prove statements (1) and (2) for the pair $(G,v)$.

Let vertex $u$ be a neighbor of $v$. We first note that statement (1) implies (2): granted (1), we have $Z_{G}(\bflam)=[Z_{G,v}(\bflam)]_{0}\cdot(1+R_{G,v})\neq 0$, since $-1\not\in K$ due to the requirement $a\in(-1,0)$. In the following, we prove statement (1) by dividing into 4 cases.

\vspace{6pt}
Case 1: $\deg_{G} v=1$ and $\deg_{G} u=1$. In this case, the edge $\{u,v\}$ is itself a connected component of $G$. Let $H$ be the graph obtained from $G$ by deleting the vertices $u$ and $v$. Let $\bflamp$ be the restriction of $\bflam$ to $V\setminus\{u,v\}$. We have 
$[Z_{G,v}(\bflam)]_{0}=(\beta+\lambda_{u})Z_{H}(\bflamp)\neq 0$,
by the induction hypothesis on $H$ and that $-\beta\not\in U$. This means $R_{G,v}$ is well defined. We then have $R_{G,v}=\dfrac{1+\gamma\lambda_{u}}{\beta+\lambda_{u}}\lambda_{v}\in f(U)\cdot U$, where $f$ denotes the M{\"o}bius transformation $r\mapsto \frac{1+\gamma r}{\beta +r}$ and $f(U)\cdot U$ denotes the Minkowski product of $f(U)$ with $U$. In order to ensure statement (1), we want $f(U)\cdot U\subset K$. Since $U$ is taken to be an arbitrarily small neighborhood of the real interval $J:=[\frac{\gamma}{\beta},1]$, the requirement on $K$ is simply $f(J)\cdot J\subset \mathrm{int}(K)$ (the interior of $K$). Using the parameter relations $\beta+\gamma>1$ and $\gamma<0$, it's easy to verify that $f(J)\cdot J$ is contained in the interval $\left(\frac{\gamma}{\beta},\frac{\beta+\gamma^{2}}{\beta^{2}+\gamma}\right]$. So $f(J)\cdot J\subset \mathrm{int}(K)$ holds if $a$ and $b$ satisfy the following two requirements: 
\begin{align}
\label{eq:req1}
a&\leq \frac{\gamma}{\beta},\\
\label{eq:req2}
b&>\frac{\beta+\gamma^{2}}{\beta^{2}+\gamma}.
\end{align}

\vspace{6pt}
Case 2: $\deg_{G}v\geq 2$ and $\deg_{G}u=1$. Let $H$ be the graph obtained from $G$ by deleting the vertex $u$ and the edge $\{u,v\}$. Let $\bflamp$ be the restriction of $\bflam$ to $V\setminus\{u\}$. Since $v$ is not an isolated vertex in $H$, we may use the induction hypothesis (1) on the pair $(H,v)$ and get $R_{H,v}\in K$. Now $[Z_{G,v}]_{0}=(\beta+\lambda_{u})[Z_{H,v}]_{0}\neq 0$, by the induction hypothesis on $(H,v)$ and that $-\beta\not\in U$. This means $R_{G,v}$ is well-defined. We then have
$$R_{G,v}=\frac{[Z_{G,v}]_{1}}{[Z_{G,v}]_{0}}=\frac{[Z_{H,v}]_{0}+\lambda_{u}\gamma[Z_{H,v}]_{0}}{\beta[Z_{H,u}]_{0}+\lambda_{u}[Z_{H,u}]_{0}}=f(\lambda_{u})\cdot R_{H,v}\in f(U)\cdot K.$$
In order to ensure statement (1), we want $f(U)\cdot K\subset K$. We impose the requirement 
\begin{equation}\label{eq:req3}
-a<b\leq 1.
\end{equation}
The requirements \eqref{eq:req1} and \eqref{eq:req3} together imply that 
$$\left\{z\in\mathbb{C}:|z|\leq \left|\frac{\gamma}{\beta}\right|\right\}\subseteq K\subseteq \{z\in\mathbb{C}:|z|\leq 1\}.$$
In particular, we have $\frac{\gamma}{\beta}\cdot K\subset K$ and by convexity of $K$, also $J\cdot K\subset K$ (recall that $J:=[\frac{\gamma}{\beta},1]$). Now, since $U$ is an arbitrarily small neighborhood of $J$, we have $f(U)\subset f(J)+W$ (the Minkowski sum of sets), where $W$ is an arbitrarily small neighborhood of 0. Using the parameter relations $\beta+\gamma>1$ and $\gamma<0$, it's easy to verify there exists a constant $0<c<1$ such that $f(J)\subset c\cdot J$. We thus have 
$$f(U)\cdot K\subset (f(J)+W)\cdot K\subset c\cdot J\cdot K+W\cdot K\subset c\cdot K+W\cdot K\subset K.$$
For the last inclusion, note that since $K$ is convex, $c\cdot K+(1-c)\cdot K=K$, so it suffices to let $W$ be sufficiently small such that $W\cdot K\subset (1-c)\cdot K$. 

\vspace{6pt}
Case 3: $\deg_{G}v=1$ and $\deg_{G}u\geq 2$. Let $H$ be the graph obtained from $G$ by deleting the vertex $v$ and the edge $\{u,v\}$. Let $\bflamp$ be the restriction of $\bflam$ to $V\setminus\{v\}$. Since $u$ is not an isolated vertex in $H$, we may use the induction hypothesis (1) on the pair $(H,u)$ and get $R_{H,u}\in K$. Now $[Z_{G,v}]_{0}=\beta[Z_{H,u}]_{0}+[Z_{H,u}]_{1}=[Z_{H,u}]_{0}\cdot (\beta+R_{H,u})\neq 0$, by the induction hypothesis on $(H,u)$ and that $-\beta\not\in K$. This means that $R_{G,v}$ is well defined. We then have
$$R_{G,v}=\frac{[Z_{G,v}]_{1}}{[Z_{G,v}]_{0}}=\frac{\lambda_{v}[Z_{H,u}]_{0}+\lambda_{v}\gamma[Z_{H,u}]_{1}}{\beta[Z_{H,u}]_{0}+[Z_{H,u}]_{1}}=f(R_{H,u})\cdot \lambda_{v}\in f(K)\cdot U.$$
In order to ensure statement (1), we want $f(K)\cdot U\subset K$. By the conformity of M\"{o}bius transformations, the map $f$ preserves orthogonality with the real line. $f$ is also decreasing as a real function on $(-\beta,+\infty)$. It follows that $f(K)$ is the disk with the real interval $[f(b),f(a)]$ as its diameter. So if we impose the requirements
\begin{equation}\label{eq:req4}
a<f(b)\text{ and }f(a)<b,
\end{equation}
then $f(K)\subset\mathrm{int}(K)$. Let $0<c<1$ be a constant so that $f(K)\subset c\cdot K$ and let $W=\{z\in\mathbb{C}:|z|\leq \varepsilon\}$ so that $U=J+W$. We thus have
$$f(K)\cdot U\subset c\cdot K\cdot (J+W)\subset c\cdot K\cdot J+c\cdot W\cdot K\subset c\cdot K+c\cdot W\cdot K\subset K.$$
For the last inclusion, since $K$ is convex, $c\cdot K+(1-c)\cdot K=K$, so it suffices to let $W$ be sufficiently small such that $c\cdot W\cdot K\subset (1-c)\cdot K$. 

\vspace{6pt}
Case 4: $\deg_{G}v\geq 2$ and $\deg_{G}u\geq 2$. Let $G_{1}$ be the graph obtained from $G$ by deleting the edge $\{u,v\}$. To shorten expressions, let $A=\left[Z_{G_{1},u,v}(\bflam)\right]_{0,0}$, $B=\left[Z_{G_{1},u,v}(\bflam)\right]_{1,0}$, $C=\left[Z_{G_{1},u,v}(\bflam)\right]_{0,1}$, and $D=\left[Z_{G_{1},u,v}(\bflam)\right]_{1,1}$. Since neither $u$ nor $v$ is an isolated vertex in $G_{1}$, we may use the induction hypothesis in similar ways as we did in \Cref{lem:diskaround0}:
\begin{itemize}
\setlength\itemsep{0pt}
\item If we define $\bflamp\in K^{V(G)}$ to be $\lambda'_{v}=0$ and $\lambda'_{x}=\lambda_{x}$ for all $x\neq v$, we would have
$R_{G_{1},u}(\bflamp)=\dfrac{\left[Z_{G_{1},u}(\bflamp)\right]_{1}}{\left[Z_{G_{1},u}(\bflamp)\right]_{0}}=\dfrac{B}{A}$. The induction hypothesis (1) gives $\frac{B}{A}\in K$.
\item If we define $\bflamp\in K^{V(G)}$ to be $\lambda'_{u}=0$ and $\lambda'_{x}=\lambda_{x}$ for all $x\neq u$, we would have
$R_{G_{1},v}(\bflamp)=\dfrac{\left[Z_{G_{1},v}(\bflamp)\right]_{1}}{\left[Z_{G_{1},v}(\bflamp)\right]_{0}}=\dfrac{C}{A}$. The induction hypothesis (1) gives $\frac{C}{A}\in K$.
\item Consider $G_{2}=G/\{u,v\}$. This means contracting the edge $\{u,v\}$: create a new vertex $w$, delete $\{u,v\}$ from $E(G)$, and in every other member of $E(G)$, substitute $w$ for any appearance of $u$ and $v$ as endpoints. If we define $\bflamp\in K^{V(G_{2})}$ to be $\lambda'_{w}=\lambda_{u}$ and $\lambda'_{x}=\lambda_{x}$ for all $x\neq w$, using the induction hypothesis on $G_{2}$, we get $\dfrac{\left[Z_{G_{2},w}(\bflamp)\right]_{1}}{\left[Z_{G_{2},w}(\bflamp)\right]_{0}}\in K$. Now, we still have $\left[Z_{G_{2},w}(\bflamp)\right]_{0}=A$, but $D=\lambda_{v}\cdot \left[Z_{G_{2},w}(\bflamp)\right]_{1}$. So $\frac{D}{A}=\lambda_{v}\dfrac{\left[Z_{G_{2},w}(\bflamp)\right]_{1}}{\left[Z_{G_{2},w}(\bflamp)\right]_{0}}\in U\cdot K$. 
\end{itemize}
To analyze the set $U\cdot K$, let $W=\{z\in \mathbb{C}:|z|\leq \varepsilon\}$, and recall that $J$ denotes the interval $[\frac{\gamma}{\beta},1]$. We have the estimation
$$
U\cdot K =(J+W)\cdot K 
\subset J\cdot K+W\cdot K 
=K+W\cdot K 
\subset K+W.
$$
The last inclusion is due to the requirement \eqref{eq:req3}.
Now, $\frac{C+\gamma D}{A}\in K+\gamma U\cdot K\subset K+\gamma K+\gamma W$. Note that the Minkowski sum of two disks is again a disk, and in particular, $K+\gamma K$ is the closed disk with the real interval $[a+\gamma b,b+\gamma a]$ as its diameter. So $$\left|\frac{C+\gamma D}{A}\right|\leq \max\{|a+\gamma b|,|b+\gamma a|\}+|\gamma|\cdot\varepsilon.$$

It follows from $\frac{B}{A}\in K$ and $-\beta\not\in K$ that $\beta A+B\neq 0$, and hence $R_{G,v}=\frac{C+\gamma D}{\beta A+B}$ is well-defined. What's more, $\left|\frac{\beta A+B}{A}\right|\geq \min_{z\in K}|\beta+z|=\beta+a$. Finally, we arrive at the estimation
$$\left|\frac{C+\gamma D}{\beta A+B}\right|=\left|\frac{C+\gamma D}{A}\right|/\left|\frac{\beta A+B}{A}\right|\leq \frac{\max\{|a+\gamma b|,|b+\gamma a|\}+|\gamma|\cdot\varepsilon}{\beta +a}.$$
We impose the final and the most crucial requirement:
\begin{equation}\label{eq:req5}
\max\{|a+\gamma b|,|b+\gamma a|\}< |a|\cdot (\beta+a).
\end{equation}
This says that when $\varepsilon$ is sufficiently small, $|R_{G,v}|=\left|\frac{C+\gamma D}{\beta A+B}\right|<|a|$. Since the disk $\{z\in\mathbb{C}:|z|\leq |a|\}$ is contained in $K$ (by requirement \eqref{eq:req3}), it follows that $R_{G,v}\in K$, concluding statement (1) and the entire induction step for the pair $(G,v)$.

\vspace{6pt}
What remains is to show that all the requirements we imposed on the constants $a$ and $b$ during the course of the proof are jointly satisfiable in the parameter range $\{(\beta,\gamma):\gamma<0\text{ and }\beta+\gamma>2-g(\beta)\}$. We defer this work to the next lemma. 
\end{proof}
\begin{lemma}
Let $g:(1,+\infty)\rightarrow(0,1)$ be the function defined in \eqref{eq:defofg}. For fixed real parameters $\beta,\gamma$ with $\gamma<0$ and $\beta+\gamma>2-g(\beta)$, there are constants $a\in(-1,0)$ and $b\in(0,\infty)$ that satisfy the requirements \eqref{eq:req1}, \eqref{eq:req2}, \eqref{eq:req3}, \eqref{eq:req4} and \eqref{eq:req5}.
\end{lemma}
\begin{proof}
Let functions $g_{1},g_{2},g_{3}:(1,+\infty)\rightarrow \mathbb{R}$ be defined by 
$$g_{1}(\beta):=\frac{\beta-2}{\beta^{2}-1},\quad g_{2}(\beta):=\frac{(\beta-1)^{2}}{\beta^{3}+\beta^{2}-\beta}\text{ and }g_{3}(\beta):=\frac{\beta^{3}+\beta^{2}-\beta}{\beta^{3}+\beta^{2}-\beta-1}.$$
So $g(\beta)=\max\{g_{1}(\beta),g_{2}(\beta)\}$. Also observe that $g_{3}(\beta)>1$ for all $\beta\in(1,+\infty)$. What's more, by direct computation, we have the relation
\begin{equation}\label{eq:g1g2g3}
3-\beta-g_{1}(\beta)=g_{3}(\beta)\cdot\left(3-\beta-g_{2}(\beta)\right).
\end{equation}

We divide the proof of the lemma into two cases.

Case 1: $\gamma\leq -1$. In this case, for $\beta$ ranging in $(1,+\infty)$, by \eqref{eq:g1g2g3} and $g_{3}(\beta)>1$,
$$
3-\beta-g_{2}(\beta)<1+\gamma\quad\Rightarrow\quad  3-\beta-g_{1}(\beta)<1+\gamma,
$$
Therefore
$$\beta+\gamma>2-\max\{g_{1}(\beta),g_{2}(\beta)\}\quad\Leftrightarrow\quad \beta+\gamma>2-g_{1}(\beta).$$
We choose $a=\frac{\gamma}{\beta}$ and $b=\max\left\{\frac{\beta+\gamma^{2}}{\beta^{2}+\gamma},-a\right\}+\varepsilon=\max\{f(a),-a\}+\varepsilon$, where $\varepsilon>0$ is a sufficiently small constant. The requirements \eqref{eq:req1}, \eqref{eq:req2}, \eqref{eq:req3} and \eqref{eq:req4} are easy to verify. Since $b>-a$ and $\gamma\leq -1$, the left hand side of \eqref{eq:req5} simplifies to $-a-\gamma b$, and the right hand side simplifies to $-a(\beta+a)$. Since the $\varepsilon$ in the formula for $b$ is arbitrarily small, \eqref{eq:req5} reduces to
$$
-\frac{\gamma}{\beta}-\gamma \max\left\{f\left(\frac{\gamma}{\beta}\right),-\frac{\gamma}{\beta}\right\}<-\frac{\gamma}{\beta} \left(\beta +\frac{\gamma}{\beta}\right),
$$
which simplifies to two inequalities:
$$\frac{\gamma}{\beta^{2}(\beta^{2}+\gamma)}(\beta-\gamma)(\beta^{2}-1)\left(\beta+\gamma-2+\frac{\beta-2}{\beta^{2}-1}\right)<0,$$
and
$$\frac{\gamma}{\beta}\left(\beta+\gamma-1+\frac{\gamma}{\beta}\right)<0.$$
The first of the two inequalities is clearly satisfied since we know that $\beta+\gamma>2-g_{1}(\beta)$. The second is satisfied because 
\begin{align*}
\beta+\gamma-1+\frac{\gamma}{\beta}&> \beta+\gamma-2+\frac{2-g_{1}(\beta)}{\beta} &(\text{since }\beta+\gamma>2-g_{1}(\beta))\\
&> \beta+\gamma-2+\frac{1}{\beta} &(\text{since }g_{1}(\beta)<1)\\
&> \beta+\gamma-2+\frac{\beta-2}{\beta^{2}-1}\\
&>0 &(\text{since }\beta+\gamma>2-g_{1}(\beta)).
\end{align*}

Case 2: $\gamma>-1$. In this case, for $\beta$ ranging in $(1,+\infty)$, by \eqref{eq:g1g2g3} and $g_{3}(\beta)>1$,
$$
3-\beta-g_{1}(\beta)<1+\gamma\quad\Rightarrow\quad  3-\beta-g_{2}(\beta)<1+\gamma,
$$
Therefore
$$\beta+\gamma>2-\max\{g_{1}(\beta),g_{2}(\beta)\}\quad\Leftrightarrow\quad \beta+\gamma>2-g_{2}(\beta).$$
We choose $a=-\frac{1}{\beta}$ and $b=f(a)+\varepsilon$, where $\varepsilon>0$ is a sufficiently small constant. \eqref{eq:req1}, \eqref{eq:req2}, \eqref{eq:req3} and \eqref{eq:req4} are easy to verify. Since $b>-a$ and $\gamma>-1$, the left hand side of \eqref{eq:req5} simplifies to $b+\gamma a$, and the right hand side simplifies to $-a (\beta+a)$. Since the $\varepsilon$ in the formula for $b$ is arbitrarily small, \eqref{eq:req5} reduces to
$$
f\left(-\frac{1}{\beta}\right)-\gamma\cdot\frac{1}{\beta}<-\frac{\gamma}{\beta} \left(\beta -\frac{1}{\beta}\right),
$$
which simplifies to $$\frac{\beta^{3}+\beta^{2}-\beta}{\beta^{2}(\beta^{2}-1)}\left(\beta+\gamma-2+\frac{(\beta-1)^{2}}{\beta^{3}+\beta^{2}-\beta}\right)>0.$$
It is satisfied since $\beta+\gamma>2-g_{2}(\beta)$.
\end{proof}
Now we are ready to prove \Cref{thm:notthreshold}. 
\begin{proof}[Proof of \Cref{thm:notthreshold}]
By symmetry between $\beta$ and $\gamma$, we may assume $\gamma<0$ and thus $\beta+\gamma>2-g(\beta)$. The theorem now follows by combining \Cref{lem:zerofreeness} with the statement (2) in  \Cref{lem:uncenteredrecur}.
\end{proof}

\section{Concluding Remarks}\label{sec:concluding}

The obvious problem left open by this work is to fully classify the complexity of approximating $Z_{G}$ in the parameter range $1\leq \beta+\gamma<2$ and (without loss of generality) $\gamma<0$. Observe that there  is an NP-hard region in this range: when $(\beta,\gamma)$ is sufficiently close to $(1,0)$, by a 2-thickening (i.e. replacing every edge by 2 parallel edges) we get a reduction from the same problem at $A=\begin{bmatrix}
\beta^{2} & 1\\ 1 & \gamma^{2}
\end{bmatrix}$, which lies in the region of ``non-uniqueness'' and is known to be NP-hard by \cite{sly2012computational}. However, this only gives us a small bounded region of NP-hardness, since the region of non-uniqueness is bounded (for a rough image, see \Cref{fig:results}). 

\Cref{thm:notthreshold} shows that in the other direction, there also exists some tractable region in the range $\{(\beta,\gamma):\gamma<0\text{ and }1\leq \beta+\gamma\leq 2\}$. Although the region where tractability is proved extends to infinity, it is rather thin (having width $g(\beta)\approx 0.1$ for small $\beta$) and its width tends to zero as $\beta\rightarrow +\infty$ (we have $g(\beta)=O(1/\beta)$). Is it possible to prove larger tractable regions?

\begin{problem}
Does there exist some $\varepsilon>0$ such that approximating $Z_{G}$ is tractable whenever $\min\{\beta,\gamma\}<0$ and $\beta+\gamma>2-\varepsilon$? 
\end{problem}

Possibly the best hope for a complete classification of approximation complexity in the range $\{(\beta,\gamma):\gamma<0\text{ and }1\leq \beta+\gamma\leq 2\}$ is to extend the uniqueness line in the positive quadrant to the negative regime. 

\begin{problem}
Is there a natural extension of the uniqueness/non-uniqueness phase transition to the case where $\min\{\beta,\gamma\}<0$?
\end{problem}

Note that our method for proving \Cref{thm:FPRAS} is to transform the problem to another problem with exclusively nonnegative parameters and use the techniques developed specifically for nonnegative problems. Interestingly, \Cref{thm:positivity} shows that the partition function is always positive in the range $\beta+\gamma\geq 1$. This points to another direction: can we reduce the problem to an ``intrinsically positive'' one?

\begin{problem}
Is it possible to transform the problem of computing $Z_{G}$ in the range $\beta+\gamma\geq 1$ to a problem with only nonnegative parameters, like the way we did in \Cref{subsec:prelimFPRAS}?
\end{problem}

\section*{Acknowledgements}
We thank Mingji Xia for many very helpful conversations about this work.

\bibliographystyle{alpha}
\bibliography{reference}
\end{document}